\documentclass[11pt]{article}
\pdfoutput=1

\usepackage{./preamble}

\title{Quantum Quasi-Monte Carlo: a window for pre-asymptotic quantum advantage}

\author[1,3]{Paolo Recchia\thanks{\tt paolo.recchia@sc.com}}
\author[1]{Zhan Yu}
\author[1]{Kelvin Koor}
\author[1,2]{Patrick Rebentrost\thanks{\tt cqtfpr@nus.edu.sg}}
\affil[1]{Centre for Quantum Technologies, NUS}
\affil[2]{School of Computing, NUS}
\affil[3]{Standard Chartered Singapore, SCB-Singapore}

\date{\today}

\begin{document}
\maketitle

\begin{abstract}
Numerical integration with Monte Carlo methods is a central computational task in many scientific and industrial applications, including financial derivative pricing and risk management. 
Classical Monte Carlo algorithms are computationally demanding: achieving an accuracy $\epsilon$ typically requires a number of function evaluations scaling as $O(1/\epsilon^2)$.
Quantum-accelerated Monte Carlo methods based on quantum amplitude estimation can in principle quadratically improve this dependence. However, \textit{quasi}-Monte Carlo methods have not been explored in the quantum context. In this work, we introduce a quantum quasi-Monte Carlo algorithm that combines low-discrepancy nets with quantum amplitude estimation. 
The proposed method prepares the quasi-random point set coherently in superposition.
The method does not yield an asymptotic improvement over classical quasi-Monte Carlo, since the total error separates into a discretization error, determined by the finite net, and a quantum estimation error. Instead, we explore a pre-asymptotic advantage window: for a target accuracy that would classically require $2^q$ low discrepancy points, one can prepare a higher-resolution net of size $2^Q$, with $Q>q$, in superposition and reach the same accuracy using significantly fewer function queries. This window can be controlled by tuning the circuit resolution and amplitude-estimation parameters, making the approach relevant for practical regimes where the number of queries is finite rather than asymptotically large.
\end{abstract}

\newpage
\tableofcontents

\newpage
\section{Introduction}
Monte Carlo methods for statistical estimation are ubiquitous in science and engineering. It is applicable when the computational problem can be reduced to the task of estimating a volume or more generally the integral of a function. Sampling points uniformly randomly can leave out significant regions in the domain of interest, especially for high-dimensional problems~\cite{glasserman2004monte}. Grid-based methods require a large number of points for high accuracy. Quasi-Monte Carlo methods have a long history of development and success in the context of estimation problems~\cite{halton1960efficiency, sobol_distribution_1967, faure1982discrepance, niederreiter1987point, niederreiter1992random, dick2010digital, glasserman2004monte}. These methods rely on sequences of deterministically generated points that provide good coverage of the domain of interest. The idea originates from Sobol's construction of point sets in the unit cube with small discrepancy~\cite{sobol_distribution_1967}, and was developed into the theory of digital nets and sequences by Niederreiter~\cite{niederreiter1987point, niederreiter1992random}, see also the textbook  \cite{dick2008exact}. 
The classical error-analysis tool is the Koksma-Hlawka inequality, which bounds the quadrature error by the product of the discrepancy of the point set and the Hardy-Krause variation of the integrand~\cite{glasserman2004monte}. With digital nets such as Sobol' sequences, for a fixed dimension $d$, and for functions of finite Hardy-Krause variation, the number of points needed to reach accuracy $\epsilon$ scales as $O(1/\epsilon)$ rather than $O(1/\epsilon^2)$, up to logarithmic factors,  competing with Heisenberg-limited quantum methods. We recommend Caflisch's review ~\cite{caflisch1998}, and Joe-Kuo~\cite{joe2008constructing} for their high-quality Sobol' constructions. 
Quasi-MC methods are used extensively in quantitative finance, in particular for pricing high-dimensional and path-dependent derivatives~\cite{glasserman2004monte}. 

Quantum superposition is shown to lead to improvements in estimating expectation values beyond the usual statistical lower bounds. The algorithmic essence of this advantage is formalized in the seminal amplitude estimation and quantum-accelerated Monte Carlo methods~\cite{brassard2000quantum, montanaro2015quantum}. These quantum methods mimic the classical method in the sense of preparing the same probability measures and evaluating the integral at the same points as conventional MC methods. The key subroutine is quantum amplitude amplification which enhances the `desired' part of a quantum state and suppresses the rest. The quantum-accelerated MC method promises to be useful for financial problems, in particular for the pricing of financial derivatives and risk analysis~\cite{montanaro2015quantum,rebentrost2018quantum,chakrabarti_threshold_2021,woerner_quantum_2019}. Detailed circuit constructions for state preparation and payoff encoding for vanilla, multi-asset and path-dependent options were shown in ~\cite{stamatopoulos_option_2020}. Because the quantum phase estimation subroutine underlying the original Quantum Amplitude Estimation (QAE) algorithm is costly on near-term hardware, several phase-estimation-free variants have since been proposed that retain the near-quadratic query scaling with substantially shallower circuits, notably Iterative QAE~\cite{grinko2021iterative} and Maximum Likelihood AE (MLAE)~\cite{suzuki2020amplitude}; see~\cite{intallura2023} for a recent survey. Resource estimates for reaching a practically useful quantum advantage in derivative pricing were given in~\cite{chakrabarti_threshold_2021}, and a quantum-accelerated multilevel Monte Carlo scheme for the underlying stochastic differential equations was proposed in~\cite{anlinden2021}. It was shown that that no net speedup over classical Monte Carlo remains when using the original Grover-Rudolph state preparation~\cite{herbert2021}, with a follow-up work  recovering some advantage again~\cite{herbert2022}. More recently, Chen, Li and Neufeld gave a complexity analysis of a quantum Monte Carlo algorithm for option pricing under general payoff functions~\cite{chenlineufeld2026}, and Shu et al.\ proposed a general-purpose quantum algorithm for numerical integration~\cite{shu2024}.

It is a natural and open question to ask how to employ the ideas of quasi-MC in a quantum algorithm and if there is the potential for advantages. We cannot expect a quantum quasi-MC algorithm to overcome the Heisenberg-limit, which is asymptotic, yet there is room to explore regimes where practical advantages could be possible. In that sense, there is room to explore \textit {pre-asymptotic} advantages, i.e., advantages not in the limiting case but rather in finite, and also realistic, regimes. Naturally, one has to venture away from asymptotic analysis, and consider numerical simulations and focus on particular regimes. 

In this work, we develop a quantum algorithm based on methods from quasi-Monte Carlo. In particular, we employ Sobol sequences in amplitude estimation. The main idea of this work can be summarized as follows. A classical qMC estimator based on $2^q$ Sobol points requires exactly $2^q$ evaluations of the integrand $f$, namely one function query for each sampled point. Our proposed QqMC algorithm generates a superposition over $2^Q=2^{q+l}$ Sobol points, where $l \in \mathbb{N}$. The number $2^Q$ of Sobol points is decoupled from the number of oracle uses of the integrand. In particular, the quantum algorithm queries the function on a high-resolution net of points, while the dependence of the QAE procedure on function queries is governed by the query budget $T$, rather than by $2^Q$ itself. The overall error naturally decomposes into two contributions: a discretization error arising from the finite Sobol approximation of the integral, and a quantum estimation error arising from the amplitude estimation procedure. We show that there are regimes where  substantially fewer queries suffice, and the width of this pre-asymptotic window can be controlled by the choice of  $Q$ and of the amplitude-estimation parameters. In our simulations, we focus on a class of convex functions that can be implemented with a reduced circuit overhead and demonstrate the pre-asymptotic advantage. \cref{fig:main} gives an overview of the key ideas of this work.

\begin{figure}
    \centering
    \includegraphics[width=1.02\linewidth]{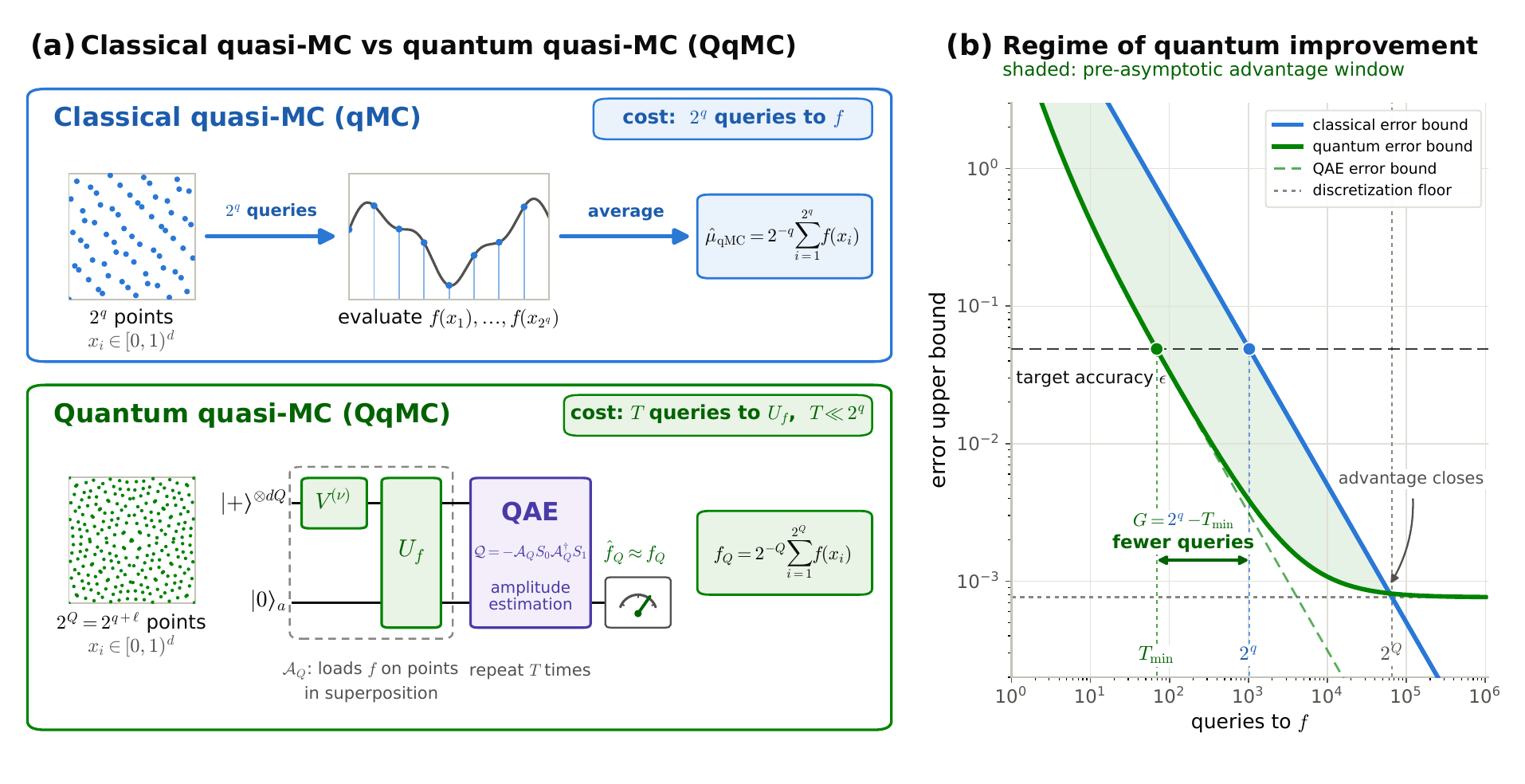}
    \caption{Panel (a). Overview of the key idea of the paper. Here, $\hat \mu_{\rm qMC}$ is the estimate with error $\leq {q^{d-1}}/{2^q}$
    Panel (b): Comparison of the qMC and QqMC bounds. 
    Classical qMC error bound is $\leq V(f) \tilde g(q) \sim {q^{d-1}}/{2^q}$.
    The QqMC error bound is $a(T) + V(f)\tilde g(Q)$. The QAE error is $a(T) = \pi/T + \pi^2/T^2$.
    The discretization floor is $V(f)\tilde g(Q)$.
    The parameters are $V(f) = 50, q = 10, l = 6, Q = q + l = 16,   \tilde g(q) = 2^{-q}$. This scenario would mean $G = 955$ fewer queries (factor $\approx 14$).
    }
    \label{fig:main}
\end{figure}

The paper is structured as follows.
In  \cref{sec:prelim} we introduce preliminaries to this work. 
In  \cref{sec:qqmc_preasymptotic}, we describe the construction of Sobol points in quantum superposition.
In  \cref{sec:QqMC},
we discuss in detail the window for pre-asymptotic advantage.
In  \cref{sec:implementation_and_numerics}, we describe the implementation, especially how to evaluate the target function on the Sobol' points and load the corresponding values into the amplitude of an ancilla qubit. 
In  \cref{sec:implementation_and_numerics}
we show our numerical implementation, and finish the paper with a discussion in  \cref{sec:concl}.

\section{Preliminaries}
\label{sec:prelim}

In this section, we present the key aspects of classical Monte Carlo and quasi-Monte Carlo theory relevant to our proposed Quantum quasi-Monte Carlo integration method. For the notation and results presented here, we mostly follow~\cite{glasserman2004monte}. For further additional details and definitions about quasi-Monte Carlo theory, we refer to \cref{sec:classical_monte_carlo}.

\subsection{Monte Carlo}

Monte Carlo methods are widely used to approximate high-dimensional integrals.
Consider a function $f:[0,1)^d\rightarrow \mathbb{R}$, and suppose that the objective is to calculate the integral
\begin{equation}
\label{eq:integ}
\mathbb{E}[f(u_1,\dots,u_d)] \coloneq \int_{[0,1)^d} f(u)\,du,
\end{equation}
where $u = (u_1,\dots,u_d)$ and the integer $d$ is the dimension. In classical Monte Carlo integration, one approximates the integral by sampling independent random points $U_i = (U_1,\dots,U_d)_i$ uniformly distributed on the unit hypercube. 
From the central limit theorem and Chebyshev’s inequality, we know that an estimator using $N$ samples  achieves the following error bound with probability at least $1 - \delta$
\begin{equation}
\label{eq:CMC}
    \epsilon_{MC} \coloneq \left| \int_{[0,1)^d} f(u)du - \frac{1}{N}\sum_{i=0}^{N-1} f(U_i) \right| \leq \frac{\sigma_f}{\sqrt{\delta N}},
\end{equation}
where $\sigma_f = \mathrm{Var}[f(U)]$, which corresponds to the well-known Monte Carlo convergence rate
\begin{equation}
\epsilon_{MC} = O(N^{-1/2}).    
\end{equation}

\subsection{Quasi-Monte Carlo}

Quasi-Monte Carlo (qMC) replaces random points with deterministic sequences chosen to fill the space more evenly than random points. Consider multidimensional points $x_i = (x^{(1)}_i, \dots, x^{(d)}_i)$ generated from a `carefully chosen' sequence \cite{glasserman2004monte}. The goal of quasi-Monte Carlo is to achieve a convergence rate close to
\begin{equation}
\label{eq:intqmc}
    \epsilon_{qMC} \coloneq \left| \int_{[0,1)^d} f(u)du - \frac{1}{N}\sum_{i=0}^{N-1} f(x_i) \right| = O\left(\frac{1}{N}\right).
\end{equation}
These results are nearly achievable by carefully and deterministically chosen sets of points, known as low-discrepancy sequences, whose definition is provided in  \cref{sec:classical_monte_carlo}. Among the most widely used and effective low-discrepancy sequences are Sobol' sequences~\cite{sobol_distribution_1967}. 

\paragraph{Sobol' sequences.}
Sobol' sequences~\cite{sobol_distribution_1967} are low-discrepancy $(t,d)$-sequences in base $2$, where $t$ is a parameter related to the quality of the sequence, often referred to as the quality parameter~\cite{dick2008exact}. For detailed definitions of low-discrepancy $(t,d)$-sequences and $(t,q,d)$-nets, as well as the meaning of the parameter $t$, we refer to  \cref{sec:classical_monte_carlo}. The main advantage of  Sobol's points is that their construction is entirely binary: the map from the binary digits of the integer index to the binary digits of the point coordinates is linear over $\mathbb{F}_2$ and can therefore be implemented using only XOR operations. This feature will be particularly convenient in the quantum-circuit construction.

Let $i \in \{ 0,1,\cdots, 2^q-1\}$ be a non-negative integer, then its binary representation is given by
\begin{equation}
i=a_0(i)+2a_1(i)+\cdots+2^{q-1}a_{q-1}(i),
\qquad a_\ell(i)\in\{0,1\},
\label{eq:sobol_index_binary}
\end{equation}
where $q$ represents the range of the integer. We collect the binary digits into the vector
\begin{equation}
\bar a_i=\bigl(a_0(i),a_1(i),\dots,a_{q-1}(i)\bigr)^\top \in \mathbb{F}_2^q.
\label{eq:sobol_a_vector}
\end{equation}

We begin with the one-dimensional construction. Let $i$ be a non-negative integer and $a_i$ its binary representation.
Let $V\in\mathbb{F}_2^{q\times q}$ be an upper-triangular invertible generator matrix, and define
\begin{equation}
\bar x_i=
\begin{pmatrix}
x_{1}(i)\\
x_{2}(i)\\
\vdots\\
x_{q}(i)
\end{pmatrix}
=
V\,\bar a_i\pmod 2.
\label{eq:sobol_linear_map}
\end{equation}
The bits $x_1(i),\dots,x_q(i)$ are then the binary digits of the $i$th Sobol' point. We highlight that $\bar x_i$ is  the binary representation of the $i$t Sobol point $x_i$,
\begin{equation}
x_i=\frac{x_1(i)}{2}+\frac{x_2(i)}{2^2}+\cdots+\frac{x_q(i)}{2^q}
=\sum_{j=1}^q \frac{x_j(i)}{2^j}.
\label{eq:sobol_point_1d}
\end{equation}

If $\bar v_j$ denotes the $j$th column of $V$, then  \cref{eq:sobol_linear_map} can be written equivalently as
\begin{equation}
\bar x_i
=
a_0(i)\bar v_1
\oplus
a_1(i)\bar v_2
\oplus\cdots\oplus
a_{q-1}(i)\bar v_q,
\label{eq:sobol_xor_form}
\end{equation}
where $\oplus$ denotes addition in $\mathbb{F}_2$, i.e., bitwise XOR. \cref{eq:sobol_xor_form} is the key computational representation: if each column $\bar v_j$ is stored as a computer word encoding the corresponding binary fraction, then the binary representation of $i$th Sobol point $\bar x_i$ is obtained by a sequence of XOR operations only. This is the form that is most useful for the circuit implementation developed later.\footnote{Alternative recursive implementations based on Gray codes are standard in classical software \cite{glasserman2004monte}, but they are not needed here. For our purposes, the linear formulation  \cref{eq:sobol_linear_map,eq:sobol_xor_form} is the most natural one.}

The essential ingredient in the Sobol' construction is therefore the choice of the generator matrix $V$, or equivalently of its columns. Arithmetic is performed over the binary field $\mathbb{F}_2$, so matrix addition is in the XOR sense. Following Sobol's original construction, the columns of $V$ are obtained from a primitive polynomial over $\mathbb{F}_2$,
\begin{equation}
p(x)=x^s+c_1x^{s-1}+\cdots+c_{s-1}x+1,
\qquad c_\ell\in\{0,1\}.
\label{eq:sobol_primitive_polynomial}
\end{equation}

Given  \cref{eq:sobol_primitive_polynomial}, one defines a sequence of odd integers $m_j$ recursively by
\begin{equation}
m_j
=
2c_1m_{j-1}
\oplus
2^2c_2m_{j-2}
\oplus\cdots\oplus
2^{s-1}c_{s-1}m_{j-s+1}
\oplus
2^s m_{j-s}
\oplus
m_{j-s},
\qquad j>s,
\label{eq:sobol_m_recurrence}
\end{equation}
and then the corresponding direction numbers are
\begin{equation}
v_j=\frac{m_j}{2^j}, \qquad j\geq 1.
\label{eq:sobol_direction_numbers}
\end{equation}
To initialize  \cref{eq:sobol_m_recurrence}, one must specify $m_1,\dots,m_s$. A minimal admissibility condition is
\begin{equation}
m_j \text{ odd}, \qquad 1\leq m_j < 2^j, \qquad j=1,\dots,s.
\label{eq:sobol_initial_values}
\end{equation}
Under  \cref{eq:sobol_initial_values}, all subsequent $m_j$ generated by  \cref{eq:sobol_m_recurrence} remain odd and satisfy $0<v_j<1$. Optimizing the initial values is an important practical issue, but it is outside the scope of the present work (for a detailed discussion of initialization and practical direction-number tables, see \cite{glasserman2004monte, dick2008exact, joe2008constructing, kuo_table} and the discussion in  \cref {sec:quality_t}.)

An immediate consequence of  \cref{eq:sobol_linear_map} is that, in one dimension, the first $2^q$ Sobol' points always form the dyadic uniform grid, up to permutation. Indeed, as $i$ ranges from $0$ to $2^q-1$, the vector of binary digits $\bar a_i$ ranges over all elements of $\mathbb{F}_2^q$. Since the leading $q\times q$ block of $V$ is upper triangular and invertible, the map $\bar a_i \mapsto \bar x_i$ is a bijection on $\mathbb{F}_2^q$. The one-dimensional Sobol' net is simply a permutation of the uniform dyadic grid. Equivalently, it is a permutation of the first $2^q$ points of the Van der Corput~\cite{niederreiter1987point} sequence in base $2$, and hence of the one-dimensional Halton~\cite{halton1960efficiency} sequence in base $2$ as well.

The multidimensional construction is obtained by assigning one generator matrix to each coordinate. More precisely, for each dimension $\nu=1,\dots,d$, choose a generator matrix
\begin{equation}
V^{(\nu)}\in\mathbb{F}_2^{q\times q},
\label{eq:sobol_generator_matrix_nu}
\end{equation}
form
\begin{equation}
\bar x_i^{(\nu)}=V^{(\nu)}\bar a_i\pmod 2,
\label{eq:sobol_coordinate_bits}
\end{equation} 
and set
\begin{equation}
x_i^{(\nu)}=\sum_{j=1}^q \frac{x_j^{(\nu)}(i)}{2^j}.
\label{eq:sobol_coordinate_value}
\end{equation}
The $i$th point in dimension $d$ is then
\begin{equation}
x_i=
\bigl(x_i^{(1)},x_i^{(2)},\dots,x_i^{(d)}\bigr)\in[0,1)^d.
\label{eq:sobol_d_dimensional_point}
\end{equation}
Different coordinates are generated from different primitive polynomials, hence from different recurrences and different sets of direction numbers.

\paragraph{Error bounds in quasi-Monte Carlo.}\label{paragraph:eb_qmc}
From the standard results in  \cref{cor:up_bound_b2_disc} and  \cref{th:khineq}, we obtain an explicit error bound for the qMC method based on a Sobol' $(t,q,d)$-net:
\begin{equation}
\label{eq:err_qmc_sob_pap}
\epsilon_{qMC} \le g V(f), \quad g \coloneq \frac{1}{2^{q- t}} \sum_{\nu=0}^{d-1} \binom{q-t}{\nu}, 
\end{equation}
which holds for any value of $q$. Note the definition of the combinatorial factor $g$ which will be adapted later. Here, $V(f)$ denotes the Hardy-Krause variation of $f$, as defined in  \cref{sec:HKvariation}. Asymptotically, for any low-discrepancy $(t,d)$-sequence we obtain (see  \cref{sec:HKvariation})
\begin{equation}
\label{eq:asymp_errqmc}
 \epsilon_{qMC}  = O \left(\frac{(\log N)^d}{N}\right),
\end{equation}
which, for low dimension $d$, scales nearly as the desired $O(1/N^{1-\varepsilon})$, where $\varepsilon > 0$ and the logarithmic factor can be absorbed into any power of $N$. 

It is worth commenting on the dependence of this upper bound on the dimension $d$. In contrast with the classical Monte Carlo bound, whose convergence rate does not explicitly depend on the dimension, the qMC upper bound exhibits a more delicate dimensional dependence which is also tightly related to the quality parameter $t$ (for a more detailed discussion on the dimension dependence refers to  \cref{sec:HKvariation} and  \cref{sec:quality_t}). 

Our main focus is the implementation of the Sobol sequence~\cite{sobol_distribution_1967} in a quantum circuit. Therefore, after the description of the implementation of Sobol points in a quantum computer,  we will also provide in  \cref{sec:quality_t} (and additional details in  \cref{app:fit_ansatz}), some insight into the behavior of the parameter $t$ in the context of Sobol sequences.

\subsection{Quantum amplitude estimation}
A $Q$-qubit register has computational basis states $\ket{i}$ indexed by integers $0 \leq i < 2^Q$. Applying Hadamard gates to the zero state prepares the uniform superposition
$
H^{\otimes Q} \ket{0} = \frac{1}{\sqrt{2^Q}} \sum_{i=0}^{2^Q-1} \ket{i}.
$
Any reversible classical computation can be embedded into a unitary operation. For a map $F$, the unitary transformation is
\[
U_F : \ket{x}\ket{y} \mapsto \ket{x}\ket{y \oplus F(x)}.
\]
When $F$ is a binary linear map, this unitary can be implemented by CNOT gates, which allow us to generate Sobol points coherently on a quantum computer.

Quantum amplitude estimation (QAE) estimates a probability encoded as the amplitude of a ``good'' subspace~\cite{brassard2000quantum}. The function under consideration is evaluated at any set of points, for example the Sobol points. After the corresponding state encoding the points has been prepared, the next step is to encode the function value into the amplitude of an ancilla qubit. Assume WLOG that
\begin{equation}
    0\le f(x_i)\le 1,
\end{equation}
where $x_i$ are the multi-dimensional Sobol points defined in \cref{eq:sobol_coordinate_value} and \cref{eq:sobol_d_dimensional_point}. We define a unitary $\mathcal A_Q$ which includes both the Sobol preparation and the amplitude-loading step:
\begin{equation}
\label{eq:AQ_state}
\mathcal A_Q\ket{ 0} = \frac{1}{\sqrt{2^Q}} \sum_{i=0}^{2^Q-1} \ket{x_i} \left(\sqrt{1-f(x_i)}\ket{0} + \sqrt{f(x_i)}\ket{1} \right).
\end{equation}
The last qubit is the flag qubit. Measuring it w.r.t. the computational basis yields $\ket{1}$ with probability
\begin{equation}
\label{eq:fQ}
f_Q = \frac{1}{2^Q} \sum_{i=0}^{2^Q-1} f(x_i),
\end{equation}
which is exactly the quasi-Monte Carlo estimator associated with the selected Sobol point set.

The implementation of the amplitude-loading map depends on the structure of $f$. In the present work, this part is treated separately in  \cref{sec:func_loading}, where we describe the arithmetic circuit used to compute a specific class of functions $f(x_i)$, or an appropriate fixed-point approximation of it, and then load it into the flag qubit. The important point for the present discussion is that the Sobol preparation and the function loading together define the state-preparation unitary $\mathcal A_Q$ in \cref{eq:AQ_state}.

Once the state \cref{eq:AQ_state} has been prepared, amplitude estimation can be applied to estimate the success probability $f_Q$. Writing
\begin{equation} 
\mathcal A_Q\ket{0} = \sqrt{1-f_Q}\,\ket{\psi_0}\ket{0} + \sqrt{f_Q}\,\ket{\psi_1}\ket{1},
\end{equation}
and parameterizing the target amplitude as
\begin{equation}
f_Q=\sin^2(\theta), \qquad \theta\in[0,\pi/2],
\end{equation}
one defines the Grover iterate
\begin{equation}
\label{eq:grover_qae}
\mathcal Q = -\mathcal A_Q S_0 \mathcal A_Q^\dagger S_1,
\end{equation}
where $S_0$ is the reflection about the all-zero input state and $S_1$ is the reflection that changes the phase of the good subspace, namely the subspace in which the flag qubit is equal to $\ket{1}$. The operator $\mathcal Q$ acts as a rotation by angle $2\theta$ on the two-dimensional invariant subspace spanned by the good and bad components.

Canonical Quantum Amplitude Estimation (QAE)~\cite{brassard2000quantum} estimates $\theta$, and therefore $f_Q$, by applying Quantum Phase Estimation (QPE) to the Grover operator $\mathcal Q$. In the standard construction, an evaluation register is first prepared in uniform superposition, controlled powers
\begin{equation}
     \mathcal Q^{2^0},\mathcal Q^{2^1},\dots, \mathcal Q^{2^{m-1}}
\end{equation}
are applied, and an inverse Quantum Fourier Transform is then performed before measurement~\cite{nielsen2010quantum}. Therefore, in canonical QAE, QPE and the inverse QFT are essential ingredients of the estimation procedure. If the total number of Grover applications is denoted by $T=\Theta(2^m)$, the canonical QAE estimator $\hat f_Q(T)$ satisfies the Brassard \emph{et al.} error bound~\cite{brassard2000quantum}
\begin{equation}
\label{eq:aerr}
|\hat f_Q(T)-f_Q| \le \frac{2\pi\sqrt{f_Q(1-f_Q)}}{T} + \frac{\pi^2}{T^2} \le  \frac{\pi}{T}+\frac{\pi^2}{T^2} =: a(T),
\end{equation}
with probability at least $8/\pi^2$. The success probability can be amplified by repeating the procedure and taking the median estimate~\cite{woerner_quantum_2019}. 

\section{Quantum quasi-Monte Carlo} 
\label{sec:qqmc_preasymptotic}

In contrast with standard Quantum Monte Carlo approaches \cite{brassard2000quantum,montanaro2015quantum,woerner_quantum_2019,rebentrost2018quantum}, in this work, we propose estimating the finite average over a deterministic low-discrepancy point set, namely, Sobol points, rather than over randomly sampled uniform points. We refer to this class of algorithms as \emph{Quantum quasi-Monte Carlo (QqMC)}. 
In this section, we provide the key aspects of the algorithm. 

\subsection{Sobol' point preparation}
\label{sec:quantum_sobol}

We now describe how the deterministic Sobol point construction is embedded into a quantum circuit and combined with QAE. The purpose of this construction is to prepare a uniform superposition over Sobol points, coherently evaluate the target function on those points, and encode the resulting quasi-Monte Carlo estimator into the success probability of an ancilla qubit.
We denote by $Q$ the number of qubits to generate $2^Q$ Sobol points per dimension. The Sobol points preparation consists of the following operations
\begin{enumerate}
\item The superposition over all Sobol indices \(i\) is prepared by applying Hadamard gates,
\begin{equation}
H^{\otimes Q}\ket{0^Q} = \frac{1}{\sqrt{2^Q}} \sum_{i=0}^{2^Q-1}
\ket{i},
\end{equation}
\item As specified in \cref{eq:sobol_linear_map}, for each dimension \(\nu\in\{1,\dots,d\}\), the Sobol construction is specified by an invertible upper-triangular binary generator matrix
\begin{equation}
V^{(\nu)}=
\begin{pmatrix}
1 & v^{(\nu)}_{0,1} & v^{(\nu)}_{0,2} & \cdots & v^{(\nu)}_{0,Q-1} \\
0 & 1 & v^{(\nu)}_{1,2} & \cdots & v^{(\nu)}_{1,Q-1} \\
0 & 0 & 1 & \cdots & v^{(\nu)}_{2,Q-1} \\
\vdots & \vdots & \vdots & \ddots & \vdots \\
0 & 0 & 0 & \cdots & 1
\end{pmatrix},
\qquad
v^{(\nu)}_{j,k}\in\{0,1\},
\end{equation}
from which we can write explicitly the Sobol points binary coefficient defined in \cref{eq:sobol_linear_map}, for each $\nu$ dimension
\begin{equation}
\begin{split}
\label{eq:y_sobol_explicit}
&y_0^{(\nu)}(i) = a_0(i)\oplus v^{(\nu)}_{0,1}a_1(i)\oplus \cdots \oplus v^{(\nu)}_{0,Q-1}a_{Q-1}(i),\\
&y_1^{(\nu)}(i) = a_1(i)\oplus v^{(\nu)}_{1,2}a_2(i)\oplus \cdots \oplus v^{(\nu)}_{1,Q-1}a_{Q-1}(i),\\
&\ \vdots \\
&y_{Q-2}^{(\nu)}(i) = a_{Q-2}(i)\oplus v^{(\nu)}_{Q-2,Q-1}a_{Q-1}(i),\\
&y_{Q-1}^{(\nu)}(i) = a_{Q-1}(i),
\end{split}
\end{equation}
Since all operations in \cref{eq:y_sobol_explicit} are XORs over \(\mathbb{F}_2\), each non-zero off-diagonal entry \(v^{(\nu)}_{j,k}=1\) is implemented by a CNOT gate with control on the qubit encoding \(a_k(i)\) and target on the qubit encoding \(a_j(i)\).
\end{enumerate}
These two steps on $(d+1)Q$ qubits register perform  the following state preparation
\begin{equation}
\ket{0^Q}\ket{0^{dQ}} \xrightarrow{\text{Hadamards + Sobol}} \frac{1}{\sqrt{2^Q}}\sum_{i=0}^{2^Q-1} \ket{i} \bigotimes_{\nu = 1}^d \ket{\bar x_i^{(\nu)}},
\end{equation}
where $\bar x^{(\nu)}_i$ are the binary coefficients of Sobol points defined in \cref{eq:sobol_linear_map}. The two-qubit gate count for dimension \(\nu\) in the above construction is exactly
\begin{equation}
N_{\mathrm{2qubits,\ Sobol}}^{(\nu)} = \#\Bigl\{(j,k)\,:\,0\le j<k\le Q-1,\ v^{(\nu)}_{j,k}=1\Bigr\},
\end{equation}
hence, in the worst case
\begin{equation}
\label{eq:2q_sobol}
N_{\mathrm{2qubits,\ Sobol}}
\le \frac{Q(Q-1)}{2},
\end{equation}
with equality only in the fully dense upper-triangular case. Summing over all dimensions,
\begin{equation}
N_{\mathrm{2qubits,\ Sobol}} = \sum_{\nu=1}^{d} N_{\mathrm{2qubits,\ Sobol}} = O(dQ^2).
\end{equation}

As for the qubit count, storing all \(d\) coordinates with \(Q\) bits each requires \(dQ\) qubits.

\subsection{Full algorithm}

In our setting, QAE is applied to the  Sobol estimator $f_Q$ in \cref{eq:fQ}. 
We highlight that the same state-preparation unitary $\mathcal A_Q$ can also be used with QAE variants that avoid the full QPE circuit, such as Iterative Amplitude Estimation (IAE)~\cite{grinko2021iterative} or Maximum-Likelihood Amplitude Estimation (MLAE)~\cite{suzuki2020amplitude}, which we adopt as the alternative for testing our method (see  \cref {subsec:implemented_iae}). These alternatives modify only the final estimation routine applied to the Grover iterate $\mathcal Q$; the Sobol preparation and the function amplitude-loading blocks remain unchanged.
In  \cref{alg:qqmc_pipeline} we present the backbone pseudo-code of our proposed QqMC method. 

\begin{algorithm}[H]
\caption{Quantum quasi-Monte Carlo protocol}
\label{alg:qqmc_pipeline}
\begin{algorithmic}[1]
\Require Dimension $d$, resolution $Q$, Sobol generator matrices $V^{(1)},\dots,V^{(d)}$, oracle for function $f$ rescaled so that $0\le f\le 1$, and an amplitude-estimation routine.
\State Prepare the uniform superposition
\[
    \ket {\chi} \coloneq \frac{1}{\sqrt{2^Q}}\sum_{i=0}^{2^Q-1}\ket{i}.
\]
\For{$\nu=1,\dots,d$}
    \State Using CNOT gates over $\mathbb{F}_2$, apply the binary Sobol map (see  \cref{sec:quantum_sobol})
    \[
        \bar x^{(\nu)}_i=V^{(\nu)}\bar a_i.
    \]
    \State The result is the transformation of bitstring $i^{(\nu)}$ into the $Q$-bit Sobol coordinate 
    \[
       \ket{i} \mapsto \ket {\bar x_i^{(\nu)}}.
    \]
\EndFor
\State Prepare the quantum arithmetic circuit for $f(x_i)$ for use on the register storing the Sobol' points $x_i$ (description in  \cref{sec:func_loading}).
\State Amplitude encode $f(x_i)$ into the ancilla qubit:
\[
    \ket{\chi} \ket{\bar x_i}\ket{0}
    \mapsto
        \ket{\chi} \ket{\bar x_i}
    \left(
    \sqrt{1-f(x_i)}\ket{0}
    +
    \sqrt{f(x_i)}\ket{1}
    \right).
\]
\State Steps 1-7  define the state-preparation unitary $\mathcal A_Q$, whose ancilla $\ket 1$ probability is
\[
    f_Q=\frac{1}{2^Q}\sum_{i=0}^{2^Q-1} f(x_i).
\]
\State Apply QAE, IAE, or MLAE to the Grover iterate (description in  \cref{subsec:implemented_iae})
\[
    \mathcal Q=-\mathcal A_QS_0 \mathcal A_Q^\dagger S_1
\]
to estimate $f_Q$.
\Ensure An estimate $\hat f_Q$ of the Sobol quasi-Monte Carlo estimator.
\end{algorithmic}
\end{algorithm}

In the query-complexity model adopted in this work, one application of $\mathcal A_Q$ or $\mathcal A_Q^\dagger$ is counted as 
one query. The Sobol point generation and the amplitude-loading circuit contribute to the gate complexity related to each query. 

\paragraph{QqMC Error.} The total error with respect to the exact integral naturally decomposes into a deterministic Sobol discretization error and a quantum amplitude-estimation error:
\begin{equation}
\label{eq:qqmc_error_decomposition}
\left|\hat f_Q(T) - \int_{[0,1)^d} f(x)\,dx \right| \le \underbrace{\left| \hat f_Q(T) - f_Q \right|}_{\mbox{\small QAE Error}} + \underbrace{\left| f_Q - \int_{[0,1)^d} f(x)\,dx \right|}_{\mbox{\small Discretization Error}}.
\end{equation}
The first term is controlled by amplitude estimation, while the second term is the quasi-Monte Carlo integration error associated with the Sobol point set.

\subsection{Heuristic for the QqMC method}
\label{sec:quality_t}

As described in  \cref {paragraph:eb_qmc}, the quality of the resulting sequence depends on the choices of primitive polynomials and initial direction numbers through the parameter $t$. Different primitive polynomials and different initial values may lead to markedly different $t$ values.

If the $d$ coordinates are generated from primitive polynomials of degrees $s_1,\dots,s_d$ (with the standard convention $s_1=0$ when the first coordinate is generated via $m_j = 1$), then Sobol's construction yields
\begin{equation}
t\coloneq t_{\rm Sobol} \coloneq\sum_{\nu=1}^d s_\nu-d + 1,
\label{eq:sobol_t_parameter}
\end{equation}
where the notation $t$ hides the dependence on $d$ and the degrees. Hence, the choice of primitive polynomials directly controls the quality parameter $t$. Sobol' showed that $t$ grows faster than $d$, but not faster than $d\log d$~\cite{sobol_distribution_1967}. In modern implementations, the direction numbers are often optimized further. In particular, the Joe-Kuo tables~\cite{joe2008constructing,kuo_table} are widely used in high-dimensional applications. 

However, from \cref{eq:err_qmc_sob_pap}, if $t > q$, the bound for the quasi-Monte Carlo integration error is undefined. It is important to distinguish between the parameter $t$ of the infinite Sobol sequence and the \emph{finite}  quality parameter $t_{\rm finite}$ of the first $2^q$ points~\cite{dick2008exact}. The former reflects the asymptotic quality of the construction as $d$ increases, whereas the latter is the relevant parameter in the $(t,q,d)$-net error bound for a fixed sequence length $2^q$. It holds that 
\begin{equation}
    t_{\rm finite} \xrightarrow{q \rightarrow \infty} t.
\end{equation}
For fixed and finite $q$, $t_{\rm finite} \le q$, so it is bounded and should not be used to infer the asymptotic growth of $t_{\rm Sobol}$. Once the degrees $s_\nu$ are fixed, the exact $t$ value is known and given by \cref{eq:sobol_t_parameter}. For the $t_{\rm finite}$ sequence instead holds~\cite{niederreiter1992random},
\begin{equation}
\label{eq:tdq}
    t_{\rm finite} = q-\rho,
\end{equation}
where $\rho$ is the linear-independence parameter of the left upper $q\times q$ submatrices of the generating matrices $V$ in \cref{eq:sobol_linear_map}. Computing $\rho$ exactly is an exact combinatorial linear-algebra problem over $\mathbb{F}_2$ whose cost grows quickly with $q$ and $d$, and whose value also depends on the specific Sobol construction through the primitive polynomials, their ordering, and the chosen initial direction numbers~\cite{marion2020algorithm,joe2008constructing,kuo_table}. To our knowledge, no general closed-form expression for $t_{\rm finite}$ is currently available. 

In the present work, we compute  $t_{\rm finite}$ explicitly for selected dimensions $d$ using the Joe-Kuo direction numbers for Sobol sequences~\cite{joe2008constructing,kuo_table}. The computation is performed up to $q=25$ and $d=50$, and the resulting values are reported in \cref{app:fit_ansatz}. Motivated by these numerical data and by the theoretical trend discussed above, we consider the following assumption.

\begin{assumption}[Empirical model for $t_{\rm finite}$]
\label{ass:t_q_empirical}
The finite quality parameter $t_{\rm finite}$ satisfies the following qualitative properties:
\begin{itemize}
    \item For each fixed dimension $d$, $t_{\rm finite}$ converges as $q \rightarrow \infty$ to a limiting parameter $t$, where $t$ is the value associated with the specific Sobol construction in~\cite{kuo_table}, as described in \cref{eq:sobol_t_parameter}.
    \item For $d=1,2$, there exist Sobol' constructions such that $t_{\rm finite} = 0$ for all $q$. In dimension $d=3$, the optimal construction yelds to $t_{\rm finite}=0$ when $q=1$ and $t_{\rm finite}=1$ for all $q\geq 2$~\cite{dick2008exact}. The~\cite{joe2008constructing} are optimal and exhibit this behavior for the special cases $d=1,2,3$.
    \item The parameter $t_{\rm Sobol}$ grows at most as
    \begin{equation}
    \label{eq:dlogd}
        t_{\rm Sobol} = O(d\log_2 d).
    \end{equation}
    \item For every finite $q$, one has
    \begin{equation}
        t_{\rm finite} \le q.
    \end{equation}
\end{itemize}

Under these assumptions, we model $t_{\rm finite}$ through the following empirical ansatz,
\begin{equation}
\label{eq:t_q_approx}
   \tilde t \coloneq  \tilde t_{\rm finite} \coloneq
    \begin{cases}
        0, & d=1,2 \quad \forall q,\\
        0, & d=3,\ q=1,\\
        1, & d=3,\ q \ge 2,\\
        \left\lceil
        \min\left(q,\,  \bigl(1-e^{-\frac{\gamma q}{t_{\rm Sobol}}}\bigr)t_{\rm Sobol}\right)
        \right\rceil,
        & d \ge 4 \quad \forall q,
    \end{cases}
\end{equation}
where $\lceil \cdot \rceil$ denotes rounding up to the integer. The parameter $\gamma$ is a fitting parameter obtained from the values reported in  \cref{tab:tdq_values}; additional details on the fitting procedure are given in  \cref{app:fit_ansatz}.
\end{assumption}

We emphasize that  \cref{eq:t_q_approx} is a heuristic model that captures the computed tables fairly well. Its role is to provide a practical approximation of the quality parameter in regimes where explicit tabulation is unavailable, while preserving $\tilde t\rightarrow t_{\rm Sobol}$ and $ \tilde t\le q$. Recalling \cref{eq:err_qmc_sob_pap}, the rigorous quasi-Monte Carlo error bound turns into the empirical bound
\begin{equation}
\label{eq:final_err_qmc_sob}
    \tilde \epsilon_{qMC} \le \tilde g V(f), \quad \tilde g \coloneq \frac{1}{2^{q-\tilde t}} \sum_{i=0}^{d-1} \binom{q-\tilde t}{i},
\end{equation}
where  $t_{\rm finite}$ was replaced with the empirical ansatz $\tilde t$. 
The resulting error bound should be interpreted as a heuristic approximation.

\section{Analysis of Quantum quasi-Monte Carlo}
\label{sec:QqMC}

Our proposed method constructs a quantum circuit that coherently generates Sobol points and estimates the integral \cref{eq:integ} by means of Quantum Amplitude Estimation (QAE). An implication of \cref{eq:qqmc_error_decomposition} is as follows: as $T\to\infty$, only the amplitude-estimation contribution vanishes, while the discrepancy term associated with the fixed net of size $2^Q$ remains. In other words, for a fixed low-discrepancy net, our QqMC method cannot asymptotically achieve a better precision than classical qMC.
As discussed, a classical qMC estimator must explicitly evaluate the integrand at all  Sobol points. In contrast, the quantum circuit prepares the corresponding function values on the Sobol net coherently with  $O(1)$ oracle queries and then extracts the summation with around $T$ uses of that circuit. Hence, 
we want to identify a pre-asymptotic window of $(T,q,l)$ in which the upper bound of the proposed Quantum Quasi-Monte Carlo estimator is smaller than the corresponding upper bound of classical qMC. Here, $T$ denotes the number of Grover-operator queries used by QAE, while $2^q$ denotes the number of Sobol points used by the reference classical qMC estimator.

As we show below, the key point is that the resulting finite estimator $f_Q$ may be significantly more accurate than the classical estimator based on $2^q$ points. Moreover, the quantum method may use substantially fewer function queries than those required by classical qMC to evaluate the same number of Sobol points.
Recall that  
\begin{equation}
\label{eq:basic_qmc_notation}
Q=q+l,
\end{equation}
where $l\in\mathbb N$ denotes the number of additional Sobol qubits used in the quantum construction. 
We also recall $a(T)$, the QAE contribution  to the error defined in \cref{eq:aerr}.
Recall the Sobol qMC bound factor $\tilde g(q)$ in \cref{eq:final_err_qmc_sob} as a function of $q$.
Consider the set of parameters $(q,l,d,V(f))$. The subsequently defined quantities are a function of the parameters in this set. 
The comparison between the qMC and QqMC error upper bounds is controlled by
\begin{equation}
\label{eq:Deltaq_def}
\boxed{
\Delta\coloneq V(f)\bigl(\tilde g(q)-\tilde g(q+l)\bigr).}
\end{equation}
This quantity measures the improvement in the error bound obtained by moving from $2^q$ to $2^{q+l}$ Sobol points. When $\Delta>0$, solving $\Delta = a(T)$ for $T$ obtains, rounded up to the next integer (for a more formal proof refer to the \cref{th:Tthreshold}),
\begin{equation}
\label{eq:Tmin_def}
\boxed{T_{\min}
\coloneq 
\left \lceil\frac{2\pi}{\sqrt{1+4\Delta}-1}\right \rceil.}
\end{equation}

Thus, $T_{\min}$ is the smallest integer $T$ satisfying $a(T)<\Delta$. 
The QqMC error upper bound is strictly lower than the classical qMC error upper bound with fewer quantum than classical function queries whenever $T$ and $q$ are such that
\begin{equation}
\label{eq:Threshold_queries}
\Delta>0
\qquad\text{and}\qquad
T_{\min}\le T\le 2^q-1.
\end{equation}

In words, $\Delta>0$ determines if there is an advantage in the error bound, and one may call this the \textit{feasibility} of an advantage. 
In addition, the difference of the number of classical queries $2^q$ and the minimum number of quantum queries $T_{\min}$ measures how many fewer queries the quantum algorithms requires. One may call this the \textit{query-count gap}
\begin{equation}
\label{eq:G_def}
\boxed{G\coloneq 2^q-T_{\min}}
\end{equation}
which measures the advantage. The condition $G\ge1$ is equivalent to the condition that  $T<2^q$ satisfies \cref{eq:Threshold_queries}. We therefore study $\Delta$ first, and then analyze the dependence of $G$ on $V(f)$, $q$, $l$, and $d$.

\subsection{Entering the feasibility window}
\label{sec:feasibility_window_quantum_advantage}
Let $l >0$. The first question is for which regimes of the parameter $q$ we have $\Delta>0$. Although increasing the number of Sobol points from $2^q$ to $2^{q+l}$ should improve the discretization bound, the factor $\tilde g$ depends on the parameter $\tilde t$. The first $q$ where $\Delta>0$, defined as
\begin{equation}
\label{eq:qmin_def}
q_{\min}\coloneq \min\{q:\Delta>0\},
\end{equation}
can be written explicitly in terms of $l, \tilde t$ and $d$, as
\begin{equation}
\label{eq:qmin_exact}
q_{\min} = 
\begin{cases} 
m_d - l & \text{if } m_d - l \ge 1, \\
1 & \text{otherwise}.
\end{cases}
\end{equation}
where $m_d$ denotes the smallest integer solution of $m - \tilde t(m) = d$. The derivation of \cref{eq:qmin_exact} is given in \cref{th:Delta_sign} and \cref{cor:qmin}. 

Here, \cref{eq:qmin_exact} makes the dependence on $d$ and $l$ transparent. First, going from $l$ to $l+1$ lowers the $q_{\min}$ by exactly one, until the smallest possible value $q_{\min}=1$ is reached. Second, the dependence on $d$ is encoded in $\tilde t$. Given the asymptotic Sobol parameter trend in \cref{eq:dlogd} and the heuristic ansatz in \cref{eq:t_q_approx}, the heuristic scaling for $q_{\min}$ is (derivation in \cref{lemma:m_d})
\begin{equation}
\label{eq:qmin_asymptotic_d}
q_{\min} = \frac{1}{1-\gamma}d-l+O\left(\frac{d}{\log d}\right).
\end{equation}
Since $\gamma\approx0.77215$ (see \cref{app:fit_ansatz}), $\frac{1}{1-\gamma}\approx 4.39$, so a simple approximation is
\begin{equation}
\label{eq:qmindl}
q_{\min} \approx4.39d-l.
\end{equation}
The value $l\approx4.39d$ is not necessary for a quantum advantage in general. Such a value of $l$ would be needed to push the $q_{\min}$ down to very small values.

This reflects the fact that, in higher dimension, the effective Sobol level $q-\tilde t$ must become sufficiently large before the bound factor $\tilde g(q)$ starts to decrease strictly.

It is also essential to highlight that, because of the rounding integer in the ansatz \cref{eq:t_q_approx}, plateau may occur, and consequently the condition $q\ge q_{\min}$ does not imply that $\Delta>0$ for every value of $q > q_{\min}$. Strict positivity fails whenever the effective level has a plateau of length at least $l$. Thus $q_{\min}$ should be interpreted as the first feasible value, but  feasibility may not always hold for all larger $q$s. However, increasing $l$ slightly reduces the likelihood of observing unfeasible points where $q > q_{\min}$. Within the ansatz defined in \cref{eq:t_q_approx} using $l\geq2$ is sufficient to eliminate such infeasibility points.

\subsection{Exploring the feasibility window}
\label{sec:fewer_quantum_queries_window}

We now study the query-count gap in \cref{eq:G_def}. The condition $G\ge1$ means that the QqMC error upper bound can be made strictly smaller than the qMC error upper bound using fewer quantum queries than the $2^q$ classical function evaluations required by qMC. In the following sections, we illustrate how the gap varies with respect to the algorithm parameters (the Sobol net resolution $q$ and the number of additional Sobol qubits $l$) as well as the intrinsic function parameters (the Hardy-Krause variation $V(f)$ and the dimension $d$).

\paragraph{Algorithm parameters.}\hfill

\begin{enumerate}[i.]

\item \textbf{Classical Sobol level $\bm{q}$}: For fixed $d$, $l$, and $V(f)>0$, the exact dependence of the query-count gap on $q$ follows directly from \cref{eq:Tmin_def} and \cref{eq:G_def}. The asymptotic behavior of $G(q)$ is qualitatively different in one dimension and in higher dimensions. For $d=1$, the ansatz gives $\tilde t_1(q)=0$, hence
\begin{equation}
\label{eq:Delta_d1_asymptotic}
\Delta=V(f)2^{-q}(1-2^{-l}).
\end{equation}
Therefore the continuous threshold satisfies 
\begin{equation}
\label{eq:Tmin_d1_asymptotic}
T_{\min} \sim \frac{\pi}{V(f)(1-2^{-l})}2^q \qquad (q \to \infty)
\end{equation}
Thus, in one dimension as $q$ grows, $T_{\min}$ remains of the same exponential order as $2^q$.

By contrast, for every fixed $d\ge2$, the ansatz implies that $\tilde t$ eventually saturates, hence
\begin{equation}
\label{eq:Delta_dge2_asymptotic}
\Delta
\sim
V(f)2^{-q}q^{d-1} \qquad (q \to \infty),
\end{equation}
up to constants depending on $d$ and $l$. Consequently, as $q \rightarrow \infty$
\begin{equation}
\label{eq:Tmin_ratio_q}
\frac{T_{\min}(q)}{2^q}\longrightarrow0,
\qquad
G(q)=2^q-T_{\min}(q)\sim2^q.
\end{equation}
Thus, for every fixed $d\ge2$, the minimum quantum query count becomes asymptotically negligible compared with the classical qMC query count, once the feasible window is entered.

\cref{fig:gq} shows this trend for $d=4,10$ and $l=1,2$. It also highlights two important pre-asymptotic effects. First, infeasible points with $\Delta=0$ may occur even after the first feasible level because of the plateaus described in the previous \cref{sec:feasibility_window_quantum_advantage}. Second, the gap $G(q)$ can still be negative inside the feasible window if the required threshold $T_{\min}$ is larger than the classical query count $2^q$. However, once $q$ enters the genuinely advantageous regime, the separation between the classical query budget $2^q$ and the required quantum threshold $T_{\min}$ grows very rapidly. In particular, the gap enlarges essentially exponentially with $q$, because the classical qMC cost scales as $2^q$. This effect is already visible with only a small number of additional Sobol qubits, such as $l=1$ or $l=2$.

\begin{figure}
    \centering
    \includegraphics[width=1\linewidth]{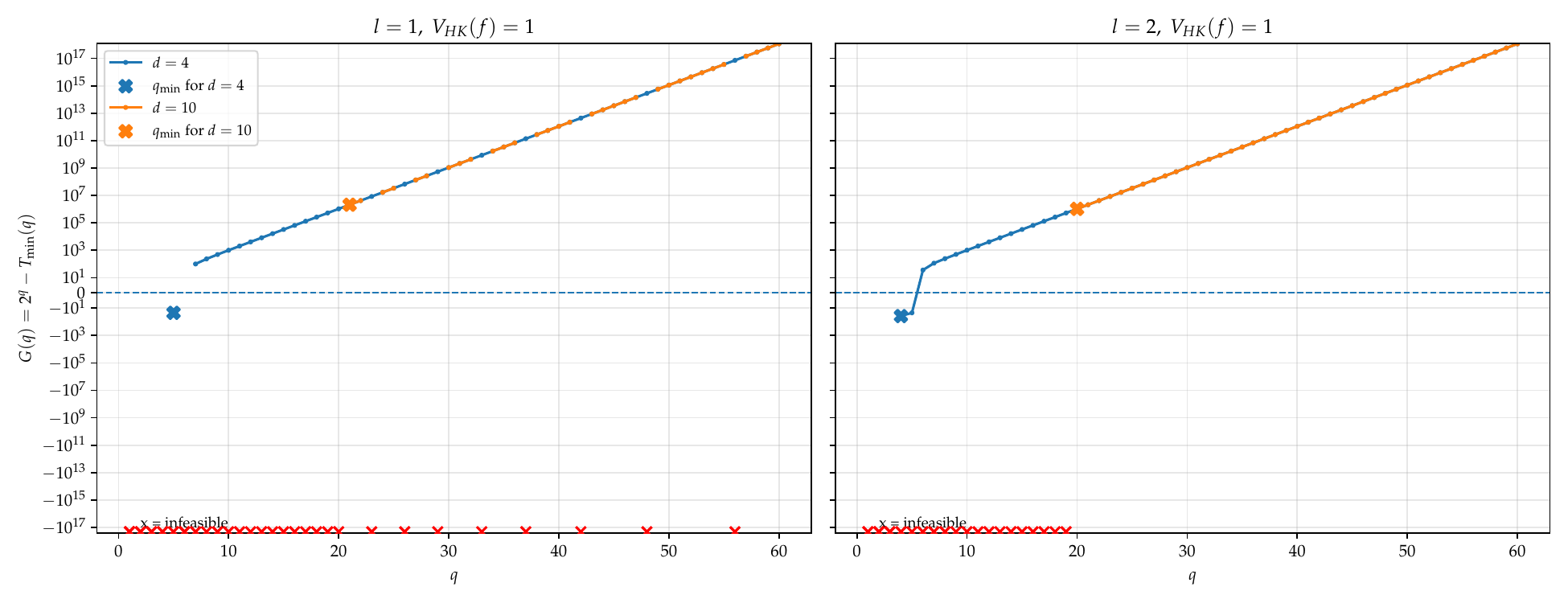}
    \caption{Gap $G(q)=2^q-T_{\min}$ as a function of $q$, for dimensions $d=4$ and $d=10$. The left panel corresponds to $l=1$, while the right panel corresponds to $l=2$. Infeasible points, namely values of $q$ for which $\Delta=0$, are highlighted in red. In the feasible regime, the gap grows rapidly with $q$, reflecting the exponential increase of the classical qMC query budget $2^q$ relative to the minimum number of quantum queries required by QqMC. This widening occurs already for a small number of additional Sobol qubits, $l=1$ or $l=2$. On the plotted scale, the exponential dependence on $q$ dominates the dimension-dependent corrections, so the curves for $d=4$ and $d=10$ become nearly indistinguishable once feasibility is reached.
    }
    \label{fig:gq}
\end{figure}

\item \textbf{Quantum Sobol level $\bm{Q}$}:
Recall that $Q=q+l$. Increasing $l$ increases the number of Sobol points $2^Q=2^{q+l}$ represented in the QqMC circuit. However, for fixed $d$, $q$, and $V(f)$, the improvement cannot continue indefinitely. Recall the definitions of $\Delta, T_{\text{min}}$ and $G$ above in \cref{eq:Deltaq_def,eq:Tmin_def,eq:G_def}. The properties for increasing $l$ are simple: $\Delta$ is nondecreasing, $T_{\min}$ is nonincreasing on the feasible window, and $G$ is nondecreasing. Moreover,
\begin{equation}
\label{eq:Deltal_limit}
\lim_{l\to \infty} \Delta=\tilde g
V(f),\qquad \lim_{l \to \infty} T_{\text{min}}
=
\left\lceil
\frac{2\pi}{\sqrt{1+4  \tilde g V(f)}-1}
\right\rceil,
\end{equation}
and the gap reaches the finite plateau $2^q-\lim_{l \to \infty} T_{\text{min}}$.
Thus, increasing $l$ is beneficial only up to a finite threshold. Beyond that point, the second term in $\Delta$ is too small to further change the integer value of $T_{\min}$, and the query-count gap remains constant. This behavior is shown in \cref{fig:gl}.

\begin{figure}
    \centering
    \includegraphics[width=0.6\linewidth]{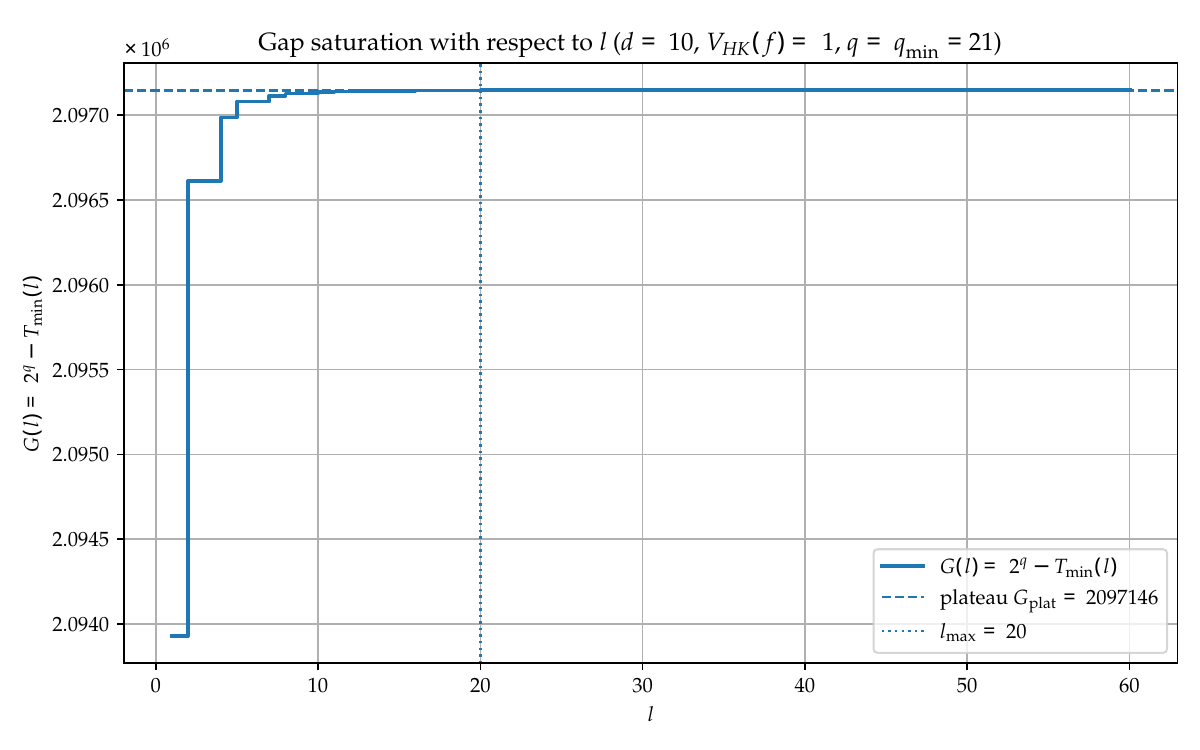}
    \caption{Gap $G(l)=2^q-T_{\min}(l)$ as a function of the number of additional qubits $l$, for fixed dimension $d=10$ and $q=q_{\min}=21$. This value of $q$ corresponds to the first strict-feasibility level when $l=1$. The plot shows that increasing $l$ initially enlarges the gap, thereby reducing the minimum number of quantum queries required to enter the advantageous regime. After a certain value of $l$, the curve reaches a clear plateau: adding more qubits further reduces the discretization contribution only marginally and no longer produces a significant increase in the query gap.
    }
    \label{fig:gl}
\end{figure}
\end{enumerate}

\paragraph{Intrinsic function parameters.}
In contrast to the previous section, in the current section we fix the algorithm parameters and investigate the dependence of the gap $G$ with respect to the function parameters (i.e. the part provided by the hypothetical user of the method). Note that we only consider the queries to the function in the quantum algorithm and not the full cost of implementing the function on the quantum computer.

\begin{enumerate}[i.]
\item \textbf{Hardy-Krause variation $\bm{V(f)}$}: From \cref{eq:Deltaq_def} we see a linear dependency on the Hardy-Krause variation $V(f)$. From this expression, we see that \cref{eq:Tmin_def} gives an inverse square root dependence. Therefore
$G(V(f))$, from \cref{eq:G_def} is a nondecreasing step function of $V(f)$. Since the integer threshold cannot go below $1$, one eventually obtains
\begin{equation}
\label{eq:TminV_integer_saturation}
\lim_{V \to \infty} T_{\min}(V)=1,
\qquad
\lim_{V \to \infty} G(V)=2^q-1.
\end{equation}
Therefore, within this upper-bound comparison, larger Hardy-Krause variation is favorable for QqMC: it widens the regime in which the QqMC method can match the classical qMC bound while using fewer queries. \cref{fig:gV} illustrates this behavior for fixed $q$ and $l$ and two different dimensions. The gap increases with $V(f)$ and eventually saturates at the maximal possible value $2^q-1$.

\begin{figure}
    \centering
    \includegraphics[width=1\linewidth]{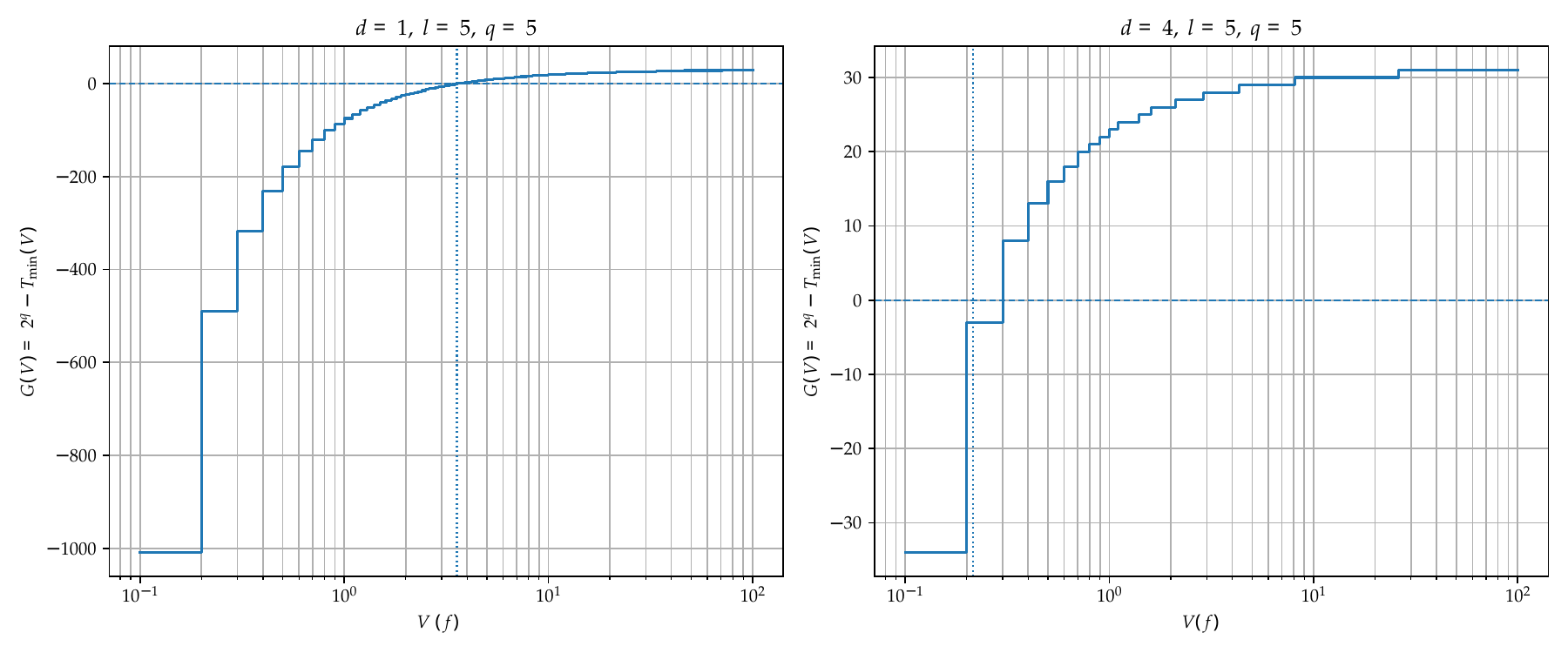}
    \caption{Gap $G(V)=2^q-T_{\min}(V)$ as a function of the Hardy-Krause variation $V(f)$, for fixed $l=5$ and $q=5$. The left panel corresponds to $d=1$, while the right panel corresponds to $d=4$. In both cases, the gap increases with $V(f)$, showing that larger variation enlarges the range of query counts for which the QqMC upper bound can remain below the classical qMC upper bound. The larger gap observed in the higher-dimensional case (right panel) suggests that, within this pre-asymptotic upper-bound comparison, achieving a quantum query advantage becomes easier as the dimension increases.
    }
    \label{fig:gV}
\end{figure}

\item \textbf{Dimension $\bm{d}$}: Studying the gap exactly as a function of $d$ would require a closed-form expression for the construction-dependent parameter $t_{\rm Sobol}$ entering the ansatz for $\tilde t$. As discussed in \cref{sec:quality_t}, such a closed form is not known in general. Therefore, the dimension-dependent conclusions below should be interpreted as consequences of the empirical ansatz in \cref{ass:t_q_empirical}, together with the scaling $t_{\rm{Sobol}}=O(d\log d)$. 

Fix $q$, $l$, and $V(f)$.
From the Joe-Kuo sequence~\cite{joe2008constructing}, we observe that the trend of $t_{\rm Sobol}$ closely follows the asymptotic scaling $d \log_2 d$. Hence, we 
simply substitute   $t_{\rm Sobol} = d \log_2 d$ into~\cref{eq:t_q_approx}.
As a consequence, $\Delta(d)$ exhibits a bell-shaped profile: it rises to a maximum at $d^\ast$ and subsequently decreases until entering the infeasible region where $\Delta(d) = 0$ at $d = d_{\max}$. 

Consequently, according to \cref{eq:Tmin_def}, $T_\text{min}(d)$ decreases monotonically up to $ d^\ast$ and then increases again until reaching the infeasible boundary at $d_{\max}$. However, for sufficiently large $q$ (provided the feasible window does not vanish), the gap between classical and quantum query complexities (the latter defined as $G(d) = 2^q - T_\text{min}(d))$ remains significant throughout the feasible interval $d \in [d_{\min}, d_{\max})$. Here, $d_{\min}$ denotes the smallest dimension satisfying the condition $T_\text{min}(d) < 2^q$. This behavior is well depicted in~\cref{fig:gd}.

\begin{figure}
    \centering
    \includegraphics[width=1\linewidth]{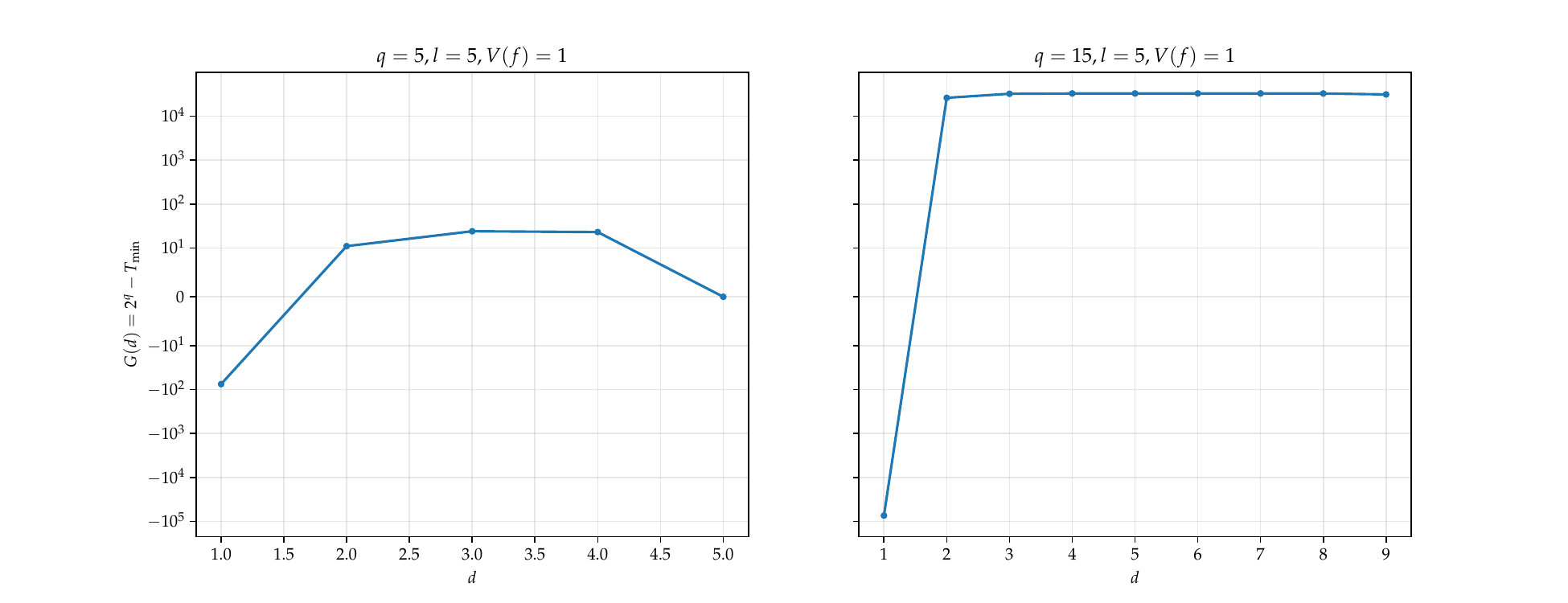}
    \caption{The gap $G(d)=2^q-T_{\min}(d)$ is plotted as a function of the dimension $d$ for fixed $l=5$ and $V(f) = 1$, with $q=5$ shown in the left panel and $q=15$ in the right. In both cases, the gap exhibits a bell-shaped trend. However, because the classical query count remains independent of $d$, increasing $q$ makes the advantage of Quantum QqMC over classical qMC significantly visible across all dimensions $d$ within the feasible window.}
    \label{fig:gd}
\end{figure}

\end{enumerate}

\section{Implementation and Numerics}
\label{sec:implementation_and_numerics}

In this section, we test the proposed Quantum-quasi Monte Carlo method on a specific class of linear integrands of the form
\begin{equation}
\label{eq:convex_linear_integrand}
f(x)=\sum_{\nu=1}^{d} w^{(\nu)} x^{(\nu)},
\qquad 
x=(x^{(1)},\dots,x^{(d)})\in [0,1)^d,
\end{equation}
where the coefficients satisfy \(w^{(\nu)}\geq 0\). In the numerical experiments considered below, we choose the weights so that \(f(x)\in[0,1)\). This class is simple enough to implement with modest overhead on a quantum simulator and illustrates the architecture of the proposed Quantum-quasi Monte Carlo method.
The goal of this section is twofold.
\begin{enumerate}
    \item First, we describe how the class of functions in \cref{eq:convex_linear_integrand} can be loaded into a quantum circuit. More precisely, given a set of \(2^Q\) Sobol points, we construct a quantum oracle whose flag-qubit amplitude encodes the estimator \(f_Q\) introduced in \cref{eq:fQ}. The loading procedure is designed so that the whole Sobol-based estimator is embedded coherently in the quantum circuit, and the resulting oracle call is counted as a single function query. Therefore, in the query model considered in this work, the circuit preparation of \(f_Q\) contributes \(O(1)\) function queries.
    \item Second, we discuss the amplitude-estimation routine used to estimate the encoded expectation value. Implementing exact Quantum Amplitude Estimation based on Quantum Phase Estimation would be computationally demanding in our setting, mainly because of the repeated controlled applications of the Grover operator and the additional overhead associated with the Quantum Fourier Transform. For this reason, inspired by the Maximum Likelihood Amplitude Estimation (MLAE)~\cite{suzuki2020amplitude,grinko2021iterative}, we implement a QPE-free version of the original QAE. We then specify how function queries are counted in this iterative protocol. In particular, we separate the intrinsic query cost associated with amplitude amplification from the additional overhead introduced by the iterative search procedure, so that the query count can be compared consistently with the one expected from the original QAE framework.
\end{enumerate}

\subsection{Loading the integrand and quantum estimator with Sobol' points}
\label{sec:func_loading}

Before entering the technical details of each block, it is useful to summarize the logic of the full circuit. The circuit acts on five groups of registers:
\begin{itemize}
\item \((d+1)Q\) qubits storing the binary coefficients of Sobol points for the \(d\) coordinates. The preparation of this register has been described in \cref{sec:quantum_sobol};
\item an \(S\)-qubit accumulator register \(\Sigma\), used to store the fixed-point approximation of \(2^S f(x_i)\);
\item an \(S\)-qubit register \(r\), later used to prepare a uniform random integer in \(\{0,\dots,2^S-1\}\), necessary for the comparator with $\Sigma$;
\item an ancillary equality register of size \(S+1\), required by the reversible comparator;
\item one flag qubit, which will encode the Bernoulli outcome associated with \(f(x_i)\).
\end{itemize}

Starting from the all-zeros state, the first step prepares a uniform superposition over the integer labels \(i\in\{0,\dots,2^Q-1\}\) and then applies the Sobol mapping for all dimensions (\cref{sec:quantum_sobol}). This step is given by
\begin{equation}
\begin{split}
& \ket{0^Q} \ket{0^{dQ}}\ket{0^S}_{\Sigma}\ket{0^S}_{r}\ket{0^{S+1}}_{\mathrm{eq}}\ket{0}_{\mathrm{flag}}
\;\xrightarrow{\;\text{Hadamards + Sobol}\;} \\
&  \left( \frac{1}{\sqrt{2^Q}}\sum_{i=0}^{2^Q-1} \ket{i} \bigotimes_{\nu = 1}^d \ket{\bar x^{(\nu)}_i} \right)  \ket{0^S}_{\Sigma}\ket{0^S}_{r}\ket{0^{S+1}}_{\mathrm{eq}}\ket{0}_{\mathrm{flag}}.
\end{split}
\end{equation}
Here \(\bar x^{(\nu)}_i\) denotes Sobol binary coefficients associated with the \(i\)th point and $\nu$ dimension (see \cref{sec:quantum_sobol}). For simplicity, we can define naturally
\begin{equation}
\bigotimes_{\nu = 1}^d \ket{\bar x^{(\nu)}_i}  =: \ket{\bar x_i}.
\end{equation}

The second step evaluates the integrand in fixed-point encoding. More precisely, the controlled additions map the Sobol registers into an \(S\)-qubit approximation 
\(\widetilde f(x_i)\approx 2^S f(x_i)\), stored in the adder register:
\begin{equation}
\label{eq:contr_adder}
\begin{split}
& \ket{i}\ket{\bar x_i}\ket{0^S}_{\Sigma}\ket{0^S}_{r}\ket{0^{S+1}}_{\mathrm{eq}}\ket{0}_{\mathrm{flag}}
\;\xmapsto{\;\text{controlled adder}\;} \\
& \ket{i} \ket{\bar x_i}\ket{\widetilde f(x_i)}_{\Sigma}\ket{0^S}_{r}\ket{0^{S+1}}_{\mathrm{eq}}\ket{0}_{\mathrm{flag}}.
\end{split}
\end{equation}

In the next step, Hadamard gates are applied to the register \(r\), producing a uniform superposition over all integers \(r\in\{0,\dots,2^S-1\}\). A reversible comparator then checks whether \(r<\widetilde f(x_i)\) and flips the flag qubit if and only if this condition is satisfied:
\begin{equation}
\label{eq:comparator}
\begin{split}
&\ket{i} \ket{\bar x_i}\ket{\widetilde f(x_i)}_{\Sigma}\ket{0^S}_{r}\ket{0^{S+1}}_{\mathrm{eq}}\ket{0}_{\mathrm{flag}}
\;\xmapsto{\;H^{\otimes S}\text{ on }r\;+\;\text{comparator}\;} \\
& \ket{i}\ket{\bar x_i}\ket{\widetilde f(x_i)}_{\Sigma}\frac{1}{\sqrt{2^S}}\sum_{r=0}^{2^S-1}
\ket{r}\ket{0^{S+1}}_{\mathrm{eq}}|[r<\widetilde f(x_i)]\rangle_{\mathrm{flag}}\\
&=
\ket{i}\ket{\bar x_i}\ket{\widetilde f(x_i)}_{\Sigma}\ket{0^{S+1}}_{\mathrm{eq}} \frac{1}{\sqrt{2^S}}\left(\sum_{r: \,r<\widetilde f(x_i)}\ket{r}\ket{1}_{\mathrm{flag}} + \sum_{r: \,r\geq\widetilde f(x_i)}\ket{r}\ket{0}_{\mathrm{flag}}\right)=: \ket{\ref{eq:comparator}}.
\end{split}
\end{equation}

The second step converts the fixed-point value stored in the accumulator into an amplitude encoding. Indeed, the register \(r\) is uniformly distributed over \(2^S\) integers, and therefore the flag is equal to \(1\) for exactly \(\widetilde f(x_i)\) of them, in superposition with the points $x_i$. 
Note that, 
\begin{equation}
\label{eq:res_comp}
|\langle i, \bar x_i,1_{\text{flag}}|\ref{eq:comparator} \rangle |^2=\frac{\widetilde S(x_i)}{2^S}\approx f(x_i).
\end{equation}

In the next sections, we describe step by step how to build a quantum circuit that evaluates \(f(x)\) on Sobol points in superposition. We also keep track of the corresponding overhead in terms of qubit count and two-qubit gate count.

\paragraph{Weighted-sum loading.} We now encode the quantity \(2^S f(x_i)\) into an \(S\)-qubit accumulator register. Let
\begin{equation}
i_\nu=\sum_{k=0}^{Q-1} y_k^{(\nu)}(i)\,2^{Q-1-k},
\qquad
x^{(\nu)}_i=\frac{i^{(\nu)}}{2^Q},
\end{equation}
where \(y_k^{(\nu)}(i)\) are the binary coefficients of the $i$th Sobol integer in dimension \(\nu\), the ones specified in \cref{eq:sobol_linear_map} and \cref{eq:y_sobol_explicit}. Then
\begin{equation}
2^S f(x_i) = 2^S \sum_{\nu=1}^{d} w^{(\nu)} x^{(\nu)}_i = 2^{S-Q}\sum_{\nu=1}^{d} w^{(\nu)} y^{(\nu)}(i) = \sum_{\nu=1}^{d}\sum_{k=0}^{Q-1}y_k^{(\nu)}(i)\,w^{(\nu)}\,2^{S-1-k}.
\end{equation}
In the present proof-of-principle implementation, the modular addition
\begin{equation}
\ket{s} \mapsto \ket{(s+c) \bmod{2^S}}
\end{equation}
is realized as an exact lookup-table permutation unitary acting on the full accumulator register, implemented as a controlled multi-qubit unitary~\cite{pennylane_controlledqubitunitary}. More scalable alternatives include structured ripple-carry and optimized fault-tolerant adders~\cite{vedral1996quantum,cuccaro2004new,gidney2018halving}.
Since \(\sum_{\nu=1}^{d} w^{(\nu)} =1\) and \(x^{(\nu)}_i\in[0,1)\), the exact value \(2^S f(x_i)\) lies in \([0,2^S]\), so in the present setting no carry beyond the \(S\)-qubit accumulator is required. Because the adder takes an integer constant as input, each real coefficient \(w^{(\nu)} 2^{S-1-k}\) must be approximated by an integer. In the current implementation, this is done by
\begin{equation}
\label{eq:ck}
c_k^{(\nu)} = \lfloor w^{(\nu)} 2^{S-1-k}\rceil,
\end{equation}
where $\lfloor \cdot \rceil$ indicates the nearest integer. So the circuit loads the fixed-point approximation
\begin{equation}
\label{eq:tilde_S}
\widetilde S(\bar x_i) = \sum_{\nu=1}^{d}\sum_{k=0}^{Q-1} y_k^{(\nu)}(i)\,c_k^{(\nu)}.
\end{equation}
Consequently,
\begin{equation}
\frac{\widetilde S(\bar x_i)}{2^S}\approx f(x_i).
\end{equation}

Hence the number of controlled additions is at most
\begin{equation}
N_{\mathrm{add}}^{(\mathrm{macro})} \le dQ,
\end{equation}
and can be smaller in practice because some coefficients \(c_k^{(\nu)}\) may be zero.

\paragraph{Comparator.} Once the fixed-point accumulator \(\widetilde S(\bar x_i)\) has been prepared, the final step is to convert it into a Bernoulli probability. For this, we prepare a \(S\)-qubit register \(r\) in a uniform superposition over \(\{0,\dots,2^S-1\}\), and compare it with the value \(\widetilde S(\bar x_i)\).

The comparator is implemented as a reversible bitwise comparison circuit and flips the flag qubit if and only if
\begin{equation}
r<\widetilde S(x_i).
\end{equation}

At the logical level, the comparator used here contains
\begin{equation}
N_{\mathrm{cmp}}^{(\mathrm{macro})}=5S
\end{equation}
multi-controlled NOT blocks, together with \(O(S)\) additional single-qubit \(X\) gates. The comparator construction used here employs an equality ancilla chain of length $S+1$. These ancillas store the prefix-equality conditions between the two registers being compared, which are needed to determine whether lower-order bits can affect the comparison result. Once the comparison flag has been written, these temporary bits are uncomputed and restored to their initial state. We therefore count the comparator cost in two parts: the temporary ancilla qubits required by this comparator design, and the number of logical multi-controlled blocks, which can then be decomposed into elementary one- and two-qubit gates according to the chosen compilation scheme.

\paragraph{Overall resource count.} The main registers are:
\begin{itemize}
\item \((d+1)Q\) qubits for the \(d\) Sobol integers and indices;
\item \(S\) qubits for the accumulator register;
\item \(S\) qubits for the random threshold register;
\item \(S+1\) ancilla qubits for the reversible comparator;
\item \(1\) flag qubit.
\end{itemize}
Hence the total qubit count is
\begin{equation}
N_{\mathrm{qubits}}
=
(d+1)Q + S + S + (S+1) + 1
=
(d+1)Q + 3S + 2.
\end{equation}

For the two-qubit gate count, it is convenient to separate the three blocks:
\begin{equation}
N_{2\mathrm{qubits}} = N_{\mathrm{Sobol}} + N_{\mathrm{add}}^{(2\mathrm{q})} + N_{\mathrm{cmp}}^{(2\mathrm{q})}.
\end{equation}
The Sobol register contributes as in \cref{eq:2q_sobol}. For the weighted-sum block, the number of logical controlled additions satisfies
\begin{equation}
N_{\mathrm{add}}^{(\mathrm{macro})} \le dQ,
\end{equation}
but the corresponding elementary two-qubit cost depends on how the controlled adder is decomposed. Likewise, for the comparator,
\begin{equation}
N_{\mathrm{cmp}}^{(\mathrm{macro})}=5S,
\end{equation}
but each multi-controlled NOT must be further decomposed into elementary gates. Therefore, the exact two-qubit count of the arithmetic and comparison layers is implementation dependent, whereas the logical-block count reported above is fixed by the circuit structure.

In the present proof-of-principle implementation, the fixed-point accumulation is the dominant block because each controlled addition is inserted as a full controlled unitary on the accumulator register. For small \(S\) and \(Q\), this is perfectly adequate for simulation and for demonstrating the feasibility of the quantum quasi-Monte Carlo method. For larger \(S\) and \(Q\), however, the arithmetic block should be replaced by a structured quantum adder in order to recover a polynomial gate scaling.

\subsection{Implemented amplitude estimation routine}
\label{subsec:implemented_iae}

We now describe the amplitude estimation routine used in the numerical experiments. The goal is to estimate an amplitude defined in \cref{eq:fQ} prepared by a unitary operator \(\mathcal A_Q\) defined in \cref{eq:AQ_state}. In our implementation, \(\mathcal A_Q\) is the Sobol plus function loading circuit described in the previous sections. For the linear integrands considered here, the flag amplitude encodes the normalized Sobol estimator. If \(c_w=\sum_{\nu=1}^d w^{(\nu)}\), then the circuit estimates
\begin{equation}
a_Q = \frac{1}{c_w} \frac{1}{2^Q} \sum_{i=1}^{2^Q} f(x_i),
\end{equation}
and the final estimate of the integral is recovered as
\begin{equation}
\hat f_Q = c_w  \hat a_Q.
\end{equation}
After applying \(k\) Grover iterates defined in \cref{eq:grover_qae}, the probability
of measuring the flag qubit in state \(\ket{1}\) is
\begin{equation}
p_k(\theta) = \sin^2\left((2k+1)\theta\right).
\label{eq:pk_theta_iae_impl}
\end{equation}

Instead of implementing the original Quantum Amplitude Estimation algorithm
based on Quantum Phase Estimation, which would require controlled powers of
\(\mathcal Q\) and an inverse Quantum Fourier Transform, we use a simplified
iterative routine inspired by QPE-free amplitude-estimation methods~\cite{suzuki2020amplitude,grinko2021iterative}. More
precisely, for a fixed schedule of Grover powers
$\mathcal K=\{k_1,\dots,k_m\}$,
we run the circuits
\begin{equation}
\mathcal Q^{k_j}\mathcal A_Q\ket{0}, \qquad j=1,\dots,m,
\end{equation}
and estimate the corresponding flag probabilities
\begin{equation}
\hat p_{k_j} = \frac{n_{k_j}}{N_{\mathrm{shots}}},
\end{equation}
where \(n_{k_j}\) is the number of observed flag outcomes equal to \(1\).
The angle \(\theta\) is estimated by a simple weighted least-squares fit. Given the domain $\Theta= [0,\pi/2]$, define the loss function
\begin{equation}
\label{eq:loss}
\mathcal L(\theta) = \sum_{k_j\in\mathcal K} (2k_j+1)^2 \left[ \hat p_{k_j} - \sin^2\left((2k_j+1)\theta\right) \right]^2,
\end{equation}
where the factor $(2k_j+1)^2$ reflects the quadratic growth of the Fisher information with the Grover amplification factor~\cite{suzuki2020amplitude}. 

The estimated angle is $\hat\theta := \arg\min_{\theta\in\Theta} \mathcal L(\theta)$, and the estimated amplitude is
\begin{equation}
\hat a = \sin^2(\hat \theta).
\end{equation}

\paragraph{Query count convention.} There are several possible query-count conventions for this implementation. A circuit with \(k\) Grover iterates contains one initial application of \(\mathcal A_Q\), followed by \(k\) applications of \(\mathcal Q\). Since each Grover iterate contains one application of \(\mathcal A_Q\) and one application of \(\mathcal A_Q^\dagger\), the corresponding number of calls to the state preparation oracle is $2k+1$. Therefore, for a schedule \(\mathcal K\), the cost is $T_{\mathcal A_Q} = \sum_{k\in\mathcal K}(2k+1)$. If one also counts $N_{\mathrm{shots}}$-times repeated circuit executions, the cost is $T_{\mathrm{phys}} = N_{\mathrm{shots}}
\sum_{k\in\mathcal K}(2k+1)$. In the numerical comparisons, however, we report the effective cost
\begin{equation}
T_{\mathrm{eff}} = \sum_{k\in\mathcal K} k,
\end{equation}
which matches the convention used in the implementation. This convention removes the classical sampling overhead due to the iterative fitting procedure, and allows a direct comparison with the query parameter appearing in the ideal QAE bound. For the geometric schedule, as the one proposed in~\cite{suzuki2020amplitude,grinko2021iterative} $\mathcal K_r=\{1,2,4,\dots,2^r\}$, with an integer $r$, we have
\begin{equation}
T_{\mathrm{eff}} = \sum_{s=0}^{r}2^s = 2\times 2^r-1.
\end{equation}
Thus \(K:=2^r=\Theta(T_{\mathrm{eff}})\), where $K$ controls the asymptotic resolution in the same way as $T$ in QAE, up to a constant factor.

\paragraph{Expected scaling and limitations.} The original QAE algorithm satisfies the bound given in \cref{eq:aerr}, which here is the benchmark against which the empirical performance of the implemented estimator is compared.
The estimator implemented here is based on \cref{eq:pk_theta_iae_impl}. 
The resolution of \(\theta\) is controlled by the largest Grover power \(K\). In this regime, one expects
\begin{equation}
|\hat \theta-\theta| = O\left(\frac{1}{K}\right),
\end{equation}
and since
\begin{equation}
\left|\sin^2(\hat \theta)-\sin^2(\theta) \right| \leq |\hat \theta-\theta|,
\end{equation}
we obtain the heuristic scaling
\begin{equation}
|\hat a-a| = O\left(\frac{1}{K}\right) = O\left(\frac{1}{T_{\mathrm{eff}}}\right).
\end{equation}
Due to the finite shots, for a fixed
\(k\), we have an additional contribution to the empirical probabilities of 
\begin{equation}
\mathrm{Var}(\hat p_k) = \frac{p_k(1-p_k)}{N_{\mathrm{shots}}} \leq \frac{1}{4N_{\mathrm{shots}}},
\end{equation}
so we have with constant probability that
\begin{equation}
|\hat p_k-p_k| = O\left(\frac{1}{\sqrt{N_{\mathrm{shots}}}}\right).
\end{equation}
After local linearization of \(p_k(\theta)\), this induces an angular uncertainty
of order
\begin{equation}
|\hat \theta-\theta| = O\left(\frac{1}{(2k+1)\sqrt{N_{\mathrm{shots}}}} \right),
\end{equation}
provided the derivative of \(p_k(\theta)\) is not close to zero and the correct branch has been identified. This is why the finite-shot error is controlled in the experiments by using a sufficiently large value of \(N_{\mathrm{shots}}\), while the query count reported in the plots is the effective amplification budget \(T_{\mathrm{eff}}\), not the total number of sampled circuit executions.

\subsection{Numerical observation of pre-asymptotic advantage}

We now report a small set of numerical experiments comparing the classical qMC estimator with the proposed QqMC estimator. These experiments are not meant to provide a complete scaling analysis. Due to the high computational cost of classically simulating the quantum circuit, at this stage we have collected data only for small problem sizes: up to $Q=6$ in one dimension and up to $Q=4$ in two dimensions. Nevertheless, the results already show a consistent and encouraging pre-asymptotic pattern.

\paragraph{Implementation and simulation settings.} For both the one- and two-dimensional experiments, we implement the Sobol-point construction coherently in the quantum circuit and encode the corresponding function value into the amplitude of a flag qubit. In the current implementation, the function-loading block relies on a fixed-point integer approximation. In particular, we allocate $S=5$ bits of integer precision for the fixed-point representation with $d=1$ and $S=4$ in $d=2$. This choice is kept fixed across the reported experiments in order to make the comparison between the two settings as uniform as possible.

In these tests, the classical qMC estimator uses $N=2^q$ Sobol points, while the QqMC estimator uses a finer Sobol discretization with $2^Q$, where $Q=q+l$. Thus, increasing $l$ corresponds to increasing the resolution of the Sobol grid encoded in the quantum state.

The amplitude is estimated via the least-squares procedure described in \cref{subsec:implemented_iae}. To minimize the weighted loss defined in \cref{eq:loss}, we discretize the interval $[0, \pi/2]$ into $N_{\mathrm{grid}}$ points and select the value that yields the minimum loss. Additionally, we specify $N_{\mathrm{shots}}$ shots to estimate the empirical probabilities. All primary implementation parameters are summarized in \cref{tab:numerical_settings}. For the one-dimensional experiment, a uniform weight $w=1.0$ is used, whereas the two-dimensional experiment employs a convex weighting function with weights $w^{(1,2)} = (0.9, 0.1)$.

\begin{table}[t]
\centering
\begin{tabular}{c c c c c c}
\hline
Dimension $d$ & Weights & Maximum $Q$ & Fixed-point precision $S$ & $N_{\mathrm{grid}}$ Loss minimization & Shots \\
\hline
$1$ & $1.0$ & $6$ & $5$ & $2\times10^5$ & $2048$ \\
$2$ & $(0.9,0.1)$ & $4$ & $4$ & $2\times10^5$ & $2048$ \\
\hline
\end{tabular}
\caption{Summary of the numerical settings used in the preliminary simulations. The QqMC estimator uses a Sobol discretization of size $M=2^Q$, with $Q=q+l$, while the classical qMC baseline uses $N=2^q$ Sobol points. In both dimensions, the fixed-point integer precision used in the function-loading circuit is fixed to $S=5$ and the number of grid points for loss minimization are $N_{\mathrm{grid}} = 2\times10^5$. Number of shots in the circuit is fixed to $N_\mathrm{shots} = 2048$.}
\label{tab:numerical_settings}
\end{table}

All simulations are performed on a classical quantum simulator using PennyLane version $0.43.1$ with the \texttt{default.qubit} backend and Python version $3.12.3$. The simulations are run on CPU, using a 13th Gen Intel\textsuperscript{\textregistered} Core\texttrademark{} i9-13900H processor. The reported experiments are classical simulations of the proposed quantum circuit. The resulting numerical error contains three main sources of approximation: the finite Sobol discretization, the fixed-point integer representation used inside the arithmetic block, and the statistical error associated with the finite-shot amplitude-estimation routine.

The error is measured through the relative error
\begin{equation}
\label{eq:relative_error_numerics}
\varepsilon_{\mathrm{alg}}
=
\frac{
\left|\hat f_{\mathrm{alg}}-\tilde f\right|
}{
|\tilde f|
},
\qquad
\mathrm{alg}\in\{\mathrm{QqMC},\mathrm{qMC}\},
\end{equation}
where $\hat f_{\mathrm{alg}}$ denotes the estimator returned by either QqMC or classical qMC, and $\tilde f$ denotes the reference value of the integral.

\paragraph{Observed pre-asymptotic advantage.} The main observation is that as $l$ increases, the QqMC estimator achieves a lower relative error while requiring fewer Grover-operator queries than the number of classical qMC points. In other words, even in this small pre-asymptotic regime, the quantum estimator benefits from the finer Sobol discretization encoded in the amplitude, while the amplitude-estimation step is still performed with a query budget below the classical sample size.

The observed relative errors in one and two dimensions are very similar. This suggests that, at least for the small values of $Q$ currently accessible in simulation, the dominant effect is the increase in Sobol resolution. A more reliable assessment of the dimensional dependence would require simulations at larger $d$, $q$, and $Q$. This trend is illustrated in \cref{fig:1D_QqMC}. The left panel reports the one-dimensional case with $w=1.0$, while the right panel reports the two-dimensional case with weights $w^{(1,2)}=(0.9,0.1)$. In both cases, increasing the number of encoded Sobol points improves the QqMC relative error, and the quantum estimator can outperform the corresponding classical qMC estimator with a number of Grover operator queries smaller than $2^q$. This supports the pre-asymptotic picture developed in the previous sections: the finer discretization $2^Q$ can be accessed coherently by the quantum circuit, while the amplitude-estimation routine may require fewer queries than the number of classical function evaluations used by qMC at resolution $2^q$.

\begin{figure}
    \centering
    \makebox[\textwidth][c]{\includegraphics[scale=0.5]{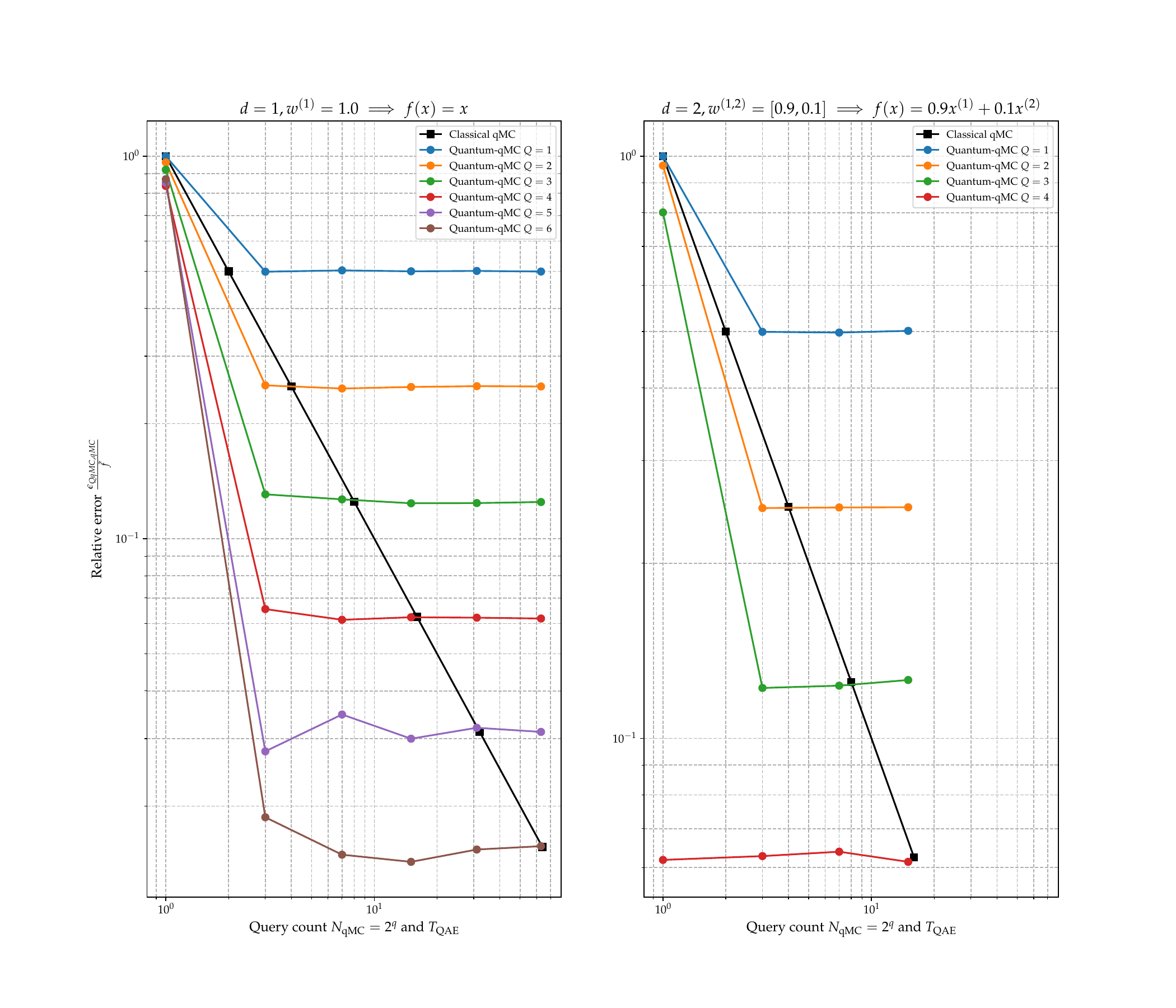}}
    \caption{Relative error $\varepsilon_{\mathrm{QqMC}}$ as the number of encoded Sobol points increases. The left panel reports the one-dimensional case with $w=1.0$, while the right panel reports the two-dimensional case with $w^{(1,2)}=(0.9,0.1)$. The parameters of the implementation are specified in \cref{tab:numerical_settings}. The results show that increasing the encoded Sobol resolution can reduce the QqMC error while keeping the number of Grover-operator queries below the number $2^q$ of classical qMC function evaluations.}
    \label{fig:1D_QqMC}
\end{figure}

At the same time, these results should be interpreted with caution. The current data cover only very small values of $d$, $q$, and $Q$, and therefore do not constitute evidence potential advantage for  other regimes, or of an asymptotic advantage. Moreover, the theoretical comparison developed in the previous sections concerns upper bounds, while the numerical results compare the actual observed errors. 

However, the implemented circuit contains several additional approximations that are not present in the idealized theoretical model, including the fixed-point integer approximation used in the arithmetic block, the grid-based least-squares fit used in the amplitude estimator, and the finite number of shots used to estimate the relevant probabilities. For this reason, the observed improvement is encouraging. Despite these practical sources of error, the QqMC estimator already reaches lower relative errors than the corresponding classical qMC estimator with fewer Grover-operator queries in the tested finite-size regimes. These preliminary results therefore suggest that the upper-bound advantage window identified theoretically may also be visible, at least qualitatively, in practical pre-asymptotic implementations.

\section{Discussion}
\label{sec:concl}

This work is to our knowledge  the first to investigate the use of low-discrepancy sequences in a quantum algorithm. The quantum algorithm is conceptually simple: it involves a subroutine for preparing the Sobol points via quantum arithmetic operations and for amplitude encoding of the function.
These subroutines can then be used together with the established techniques of amplitude estimation and its variants. The focus of the work is hence on finite parameter regimes where the quantum algorithm could produce advantages over the classical methods, quantified by the number of queries to the function. 

Our main result is the demonstration of a window of possible quantum advantage that is non-asymptotic. Using superpositions of more points can give (i) more accurate results at an equal number of queries or (ii) use less queries at the same accuracy. We have provided an analytical discussion in  \cref{sec:QqMC}. 
We note that, \cref{eq:Threshold_queries} identifies a window in the $(T, q)$-plane where the available upper bound for $\epsilon_{\rm QqMC}$ is
smaller than the corresponding upper bound for $\epsilon_{\rm qMC}$, while also requiring fewer quantum queries than the
2q classical qMC evaluations. We emphasize that this is a statement about the upper bounds and not about the
errors themselves, i.e., the existence of a window does not necessarily imply 
$\epsilon_{\rm QqMC} < \epsilon_{\rm qMC}$. Nevertheless,
the widening of the query-count gap inside this window, together with the preliminary numerical evidence
reported in \cref{sec:implementation_and_numerics}, suggests that the same region is a natural candidate for an actual quantum advantage.
Importantly, the analysis relies on \cref{ass:t_q_empirical}. This ansatz reproduces the computed values of tfinite
fairly well up to the tested dimensions and leads to a consistent qualitative picture across all the parameter
studies carried out in this section. However, since no exact closed form for $t_{\rm finite}$ is available in general,
and since the precise behavior depends on the chosen construction, the results should only be understood
as evidence for the potential for a pre-asymptotic difference in estimation error under the assumptions
considered here

The numerics demonstrate the effect on simple functions in \cref{sec:implementation_and_numerics}. There are several directions for future work. We leave more detailed study of the advantage window, in particular for specific use cases such as financial derivatives. This will involve more complex functions, where one has consider also the computational cost of quantum implementation. We also left open the exploration of a pre-asymptotic advantage for diverse types of low-discrepancy sequences other than Sobol' and for the vanilla Riemannian uniform grid. Nevertheless, the current work presents a valuable stepping stone for the use of quasi-random sequences in quantum algorithms.

\section*{Acknowledgments}
This work is supported by the National Research Foundation, Singapore through the National Quantum Office, hosted in A*STAR, under its Centre for Quantum Technologies Funding Initiative.
LLM tools were used for some low-level calculations, some of the code underlying the figures, and the initial drafting of the manuscript. 

\addcontentsline{toc}{section}{References}
\bibliographystyle{alpha}
\bibliography{ref}

\newpage
\appendix

\section{Theoretical background on Monte Carlo and quasi-Monte Carlo}
\label{sec:classical_monte_carlo}

\subsection{Nets, sequences, discrepancy and star discrepancy}
\label{sec:disc_def}

The uniformity of a point set is quantified by the discrepancy and star discrepancy, as defined below \cite{glasserman2004monte}.
\begin{definition}[Discrepancy]
\label{def:disc}
Let $\mathcal{A}$ be a family of (Lebesgue measurable) subsets of $[0,1)^d$. The discrepancy of points $\{x_0,\dots,x_{N-1}\}$ relative to $\mathcal{A}$ is
\begin{equation}
\label{eq:disc}
D(x_0,\dots,x_{N-1};A)
=
\sup_{A\in \mathcal{A}}
\left|
\frac{\#\{x_i\in A\}}{N} - \mathrm{vol}(A)
\right|.
\end{equation}
Here, $\#\{x_i \in A\}$ denotes the number of points $x_i$ contained in $A$, and $\mathrm{vol}(A)$ denotes the volume (measure) of $A$.
\end{definition}

\begin{definition}[Star discrepancy]
Restricting the same definition as in Definition~\ref{def:disc} to the case in which $\mathcal{A}$ consists of rectangles of the form
\begin{equation}
\prod_{j=1}^d [0,u_j)    
\end{equation}
defines the star discrepancy $D^*(x_0,\dots,x_{N-1})$.
\end{definition}

It is widely believed that the following conjecture holds~\cite{niederreiter1992random}:

\begin{conjecture}[Star discrepancy lower bound]
\label{conj}
For any finite sequence $x_1,\dots,x_n$ and any dimension $d$, there is a constant $c_d$ depending only on the dimension $d$ such that
\begin{equation}
    D^*(x_0,\dots,x_{N-1}) \geq c_d \frac{(\log N)^{d-1}}{N}.
\end{equation}
\end{conjecture}
Conjecture~\ref{conj} has been proved for $d\le2$~\cite{schmidt1972irregularities}, but it is widely believed to be true for any $d$~\cite{niederreiter1992random}. Conjecture~\ref{conj} motivates the following definition of a low-discrepancy sequence.

\begin{definition}[Low-discrepancy sequences]
In light of Conjecture~\ref{conj}, quasi-random sequences with the best possible order of convergence are those satisfying
\begin{equation}
    D^*(x_0,\dots,x_{N-1}) \leq C \frac{(\log N)^{d}}{N},
\end{equation}
where $C$ is a constant depending on the sequence. These sequences are called \emph{low-discrepancy} sequences.
\end{definition}

Another widely used class of low-discrepancy sequences is given by $(t,q,d)$-nets and, subsequently, $(t, d)$-sequences introduced by Niederreiter~\cite{niederreiter1987point}. Their definitions follow.

\begin{definition}[Elementary interval in base $b$]
An elementary interval in base $b$, with $b \geq 2$ an integer (also called a $b$-ary box), is a subset of $[0,1)^d$ of the form
\begin{equation}
    \prod_{i=1}^d \left[\frac{a_i}{b^{j_i}}, \frac{a_i + 1}{b^{j_i}} \right),
\end{equation}
where $j_i \in \{0,1,\dots\}$ and $a_i \in \{0,\dots,b^{j_i}-1\}$.
\end{definition}

\begin{definition}[$(t,q,d)$-nets]
\label{def:nets}
For integers $0 \leq t \leq q$, a $(t, q, d)$-net in base $b$ is a set of $b^q$ points in $[0, 1)^d$ such that exactly $b^t$ points fall in each $b$-ary box of volume $\frac{b^t}{b^q}$.
\end{definition}

\begin{definition}[$(t,d)$-sequences]
\label{def:tdseq}
A $(t,d)$-sequence in base $b$ is a sequence such that, for all $q > t$, each segment $\{x_i : jb^q < i \leq (j + 1)b^q\}$ for $j =\{0, 1, \dots\}$ is a $(t,m,d)$-net.    
\end{definition}

In other words, a $(t,d)$-sequence in base $b$ is a sequence whose successive blocks of $b^q$ points form $(t,q,d)$-nets.
It is important to note from Definition~\ref{def:nets} and Definition~\ref{def:tdseq} that smaller values of $t$ correspond to a higher degree of uniformity in the sequence. In fact, when $t$ is small, even relatively small $b$-ary boxes contain approximately the correct number of points. In the extreme case $t=0$, the definition leads to a sequence in which each admissible $b$-ary box contains exactly the prescribed number of points, corresponding to an almost perfectly uniform grid. Furthermore, all else being equal, a smaller base $b$ is preferable, since the uniformity properties of $(t,q,d)$-nets and $(t,d)$-sequences become apparent already for point sets of size $b^q$. As the base $b$ increases, a larger number of points is required before these uniformity properties manifest themselves. In this work, we mostly focus on $(t,q,d)$-nets in base $b=2$. For this purpose, we present the following relevant theorems.

\begin{theorem}[Corollary 5.3 in \cite{dick2010digital}]
\label{cor:up_bound_b2_disc}
The star discrepancy of a $(t, q, d)$-net in base $b=2$ satisfies
\begin{equation}
\label{eq:dstar_base2}
    D^*(x_0,\dots,x_{2^q-1}) \le 2^{t-q} \sum_{\nu=0}^{d-1} \binom{q-t}{\nu}.
\end{equation}
\end{theorem}

\begin{theorem}[Theorem 4.17 in~\cite{niederreiter1992random}]
\label{th:niederreiter}
If $(x_0,\dots, x_{N-1})$ is a $(t,d)$-sequence in base $b=2$, then
\begin{equation}
\label{eq:dstdisc}
D^*(x_0,\dots,x_{N-1}) \le C(d)2^t\frac{(\log N)^d}{N} + O \left( 2^t \frac{(\log N)^{d-1}}{N} \right)
\end{equation}
where $O(\cdot)$ does not depend on $N$ and $t$ and the constant prefactor $C(d)$ is given by~\cite{niederreiter1992random}:
\begin{equation}
\label{eq:low_disc_prefact}
    C(d) = \frac1{d!} \left( \frac{1}{2\log 2}\right)^d.
\end{equation}
\end{theorem}

From Theorem~\ref{th:niederreiter}, any $(t,d)$-sequence (and, as a special case, any $(t,q,d)$-net) is a low-discrepancy sequence.

\subsection{Hardy-Krause Variation and the Koksma-Hlawka Theorem}
\label{sec:HKvariation}

The definitions provided in the previous sections are essential for proving an upper bound on the error in the integral estimate in Eq.~\eqref{eq:intqmc} when using a quasi-Monte Carlo method. The error of qMC integration is bounded through the Koksma-Hlawka inequality, which requires the following preliminary definition.

\begin{definition}[Measure of variation of a function $f$]
\label{def:mesvar}
Consider an interval consisting of rectangles of the form
\begin{equation}
    J = [u_1^-,u_1^+]\times\dots \times[u_d^-,u_d^+],
\end{equation}
with $0 \le u^{-}_i \le u^{+}_i \le 1$ and  $ i = 1,\dots , d$. Let $\mathrm{Even}(J)$ be the set of vertices of $J$ with an even number of $+$ superscripts, and let $\mathrm{Odd}(J)$ contain those with an odd number of $+$ superscripts. We define
\begin{equation}
\label{eq:delta}
    \Delta(f;J) = \sum_{u \in\mathrm{Even}(J) } f(u) -  \sum_{u \in\mathrm{Odd}(J) } f(u),
\end{equation}
which represents the alternating sum of the values of $f(u)$ over the vertices of $J$, with adjacent vertices assigned opposite signs. The unit hypercube can be partitioned into a set $\mathcal{P}$ of rectangles of the form $J$. Letting $\mathcal{P}$ range over all such partitions, we define the measure of variation of a function $f$ as
\begin{equation}
\label{eq:variation}
 V^{(d)}(f) = \sup_{\mathcal{P}} \sum_{J \in \mathcal{P}} |\Delta(f;J)|.
\end{equation}
In the case in which the partial derivatives of $f$ are continuous, the measure of variation of a function $f$ can also be written as
\begin{equation}
 V^{(d)}(f) = \int_0^1 \dots \int_0^1 \left|\frac{\partial^d f}{\partial u_1 \dots \partial u_d} \right| du_1 \dots du_d.
\end{equation}
\end{definition}

\begin{definition}[Hardy-Krause variation]
\label{def:HKvar}
Let $f$ be defined on $[0,1)^d$. For any $1 \le k \le d$ and any $1 \le i_1 < i_2 < \dots < i_k \le d$, consider the function on $[0, 1]^k$ obtained by restricting $f$ to points $(u_1, \dots, u_d)$ with $u_j = 1$ if $j \notin \{i_1, \dots , i_k\}$ and with $(u_{i_1} , \dots , u_{i_k})$ ranging over all of $[0, 1]^k$. We denote by $V^{(k)}(f ; i_1, \dots ,i_k)$ the measure of variation $V^{(k)}$, as in Definition~\ref{def:mesvar}, of the restricted function defined above. The Hardy-Krause variation is
\begin{equation}
\label{eq:HKvar}
V(f) = \sum_{k=1}^d \sum_{1\le i_1<\dots<i_k\le d} V^{(k)}(f;i_1,\dots,i_k).
\end{equation}
\end{definition}

\begin{theorem}[Koksma-Hlawka inequality]
\label{th:khineq}
If $f$ has finite Hardy-Krause variation, then
\begin{equation}
\label{eq:khineq}
 \epsilon_{qMC} = \left| \int_{[0,1)^d} f(u)du - \frac{1}{N}\sum_{i=0}^{N-1} f(x_i) \right|\le V(f)D^*(x_0,\dots,x_{N-1}).
\end{equation}
\end{theorem}
The integration error is bounded by the product of two terms:
\begin{itemize}
\item a measure of the Hardy-Krause variation of the integrand $f$,
\item a measure of the uniformity of the points (the star discrepancy of the sequence).
\end{itemize}
In particular, if we choose a low-discrepancy $(t,d)$-sequence (or net), then from Eq.~\eqref{eq:khineq} and~\eqref{eq:dstdisc}, the error bound for the quasi-Monte Carlo estimator reads
\begin{equation}
\label{eq:tderrbdisc}
 \epsilon_{qMC}  \le V(f) C(d,b)b^t \frac{(\log N)^d}{N} + O\left(b^t \frac{(\log N)^{d-1}}{N} \right),
\end{equation}
which leads to the well-known Quasi-Monte Carlo bound
\begin{equation}
 \epsilon_{qMC}  = O \left(\frac{(\log N)^d}{N}\right).
\end{equation}
For low dimension $d$, the bound scales nearly as the desired $O(1/N^{1-\varepsilon})$, where $\varepsilon > 0$ and the logarithmic factor can be absorbed into any power of $N$. 

It is important to highlight the dependence of this upper bound on the dimension $d$. From Theorem~\ref{th:niederreiter}, Eq.~\eqref{eq:dstdisc}, and Eq.~\eqref{eq:low_disc_prefact}, two competing effects emerge. On the one hand, the first leading order constant $C(d,b)$ decreases superexponentially, due to its factorial dependence on $d$. On the other hand, the factor $(\log N)^d$ grows polylogarithmically with the dimension. In addition, there is also an implicit dependence on $d$ through the quality parameter $t$. In practice, keeping $t$ small becomes increasingly difficult as $d$ grows, and this behavior depends strongly on the specific family of low-discrepancy sequences under consideration. In particular, while Faure sequences~\cite{faure1982discrepance} achieve the optimal value $t=0$ at the price of using a base at least as large as the dimension, Sobol's sequences~\cite{sobol_distribution_1967} retain the computational advantages of base $2$ but typically exhibit values of $t$ that grow with the dimension (more details are in the following Section). Additionally, from Corollary~\ref{cor:up_bound_b2_disc} and Eq.~\eqref{eq:dstar_base2}, we observe that the non-asymptotic upper bound (valid for all $q$) can only increase in $d$, no matter how low the quality parameter $t$ is. Therefore, although qMC can substantially outperform classical Monte Carlo in low and moderate dimensions, its practical efficiency in higher dimensions depends critically not only on the asymptotic form of the discrepancy bound, but also on the availability of constructions with sufficiently small values of $t$.

\section{Fitting the ansatz to quality parameters generated from the Kuo direction numbers}
\label{app:fit_ansatz}

This appendix describes how the finite-resolution Sobol quality parameter $t_{\rm{finite}}$ is computed and how the empirical ansatz used in \cref{eq:t_q_approx} is fitted.

\subsection{Exact computation of $t_{\rm{finite}}$}
\label{app:exact_tq}

For each dimension $d$, we generate the Sobol generating matrices according to the primitive polynomial degrees and direction numbers of the Joe-Kuo construction~\cite{joe2008constructing}. Let $ V^{(\nu)}$ denote the binary generating matrix associated with the $\nu$th Sobol coordinate defined in \cref{eq:sobol_generator_matrix_nu}. The finite-resolution quality parameter $t_{\rm{finite}}$ is then computed from the linear-independence parameter $\rho$ as in \cref{eq:tdq} and \cite{niederreiter1992random}. The parameter $\rho$ is defined as the largest integer $\rho \le q$ such that, for every choice of non-negative integers
\begin{equation}
    k_1,\dots,k_d \in \mathbb{N}_0,
    \qquad
    k_1+\cdots+k_d \le \rho,
\end{equation}
the collection of the first $k_\nu$ rows of each matrix $V^{(\nu)}$, for $\nu=1,\dots,d$, is linearly independent over $\mathbb{F}_2$. Equivalently,
\begin{equation}
    \rho = \max \left\{ \rho \le q: \operatorname{rank}_{\mathbb{F}_2}
    \left(
    \begin{array}{c}
    V^{(1)}_{(1:k_1)} \\
    \vdots \\
    V^{(d)}_{(1:k_d)}
    \end{array}
    \right)
    = k_1+\cdots+k_d \quad \forall (k_1,\dots,k_d) \text{ with } \sum_{\nu=1}^d k_\nu \le \rho \right\},
    \label{eq:rho_definition}
\end{equation}
where $V^{(\nu)}_{(1:k_\nu)}$ denotes the submatrix formed by the first $k_\nu$ rows of $V^{(\nu)}$. In this work, the values of $t_{\rm{finite}}$ are computed exactly, by direct binary-rank checks, for $d \le 15  \cup  d=20,25,30,40,50$ and $q\le 25$. \cref{tab:tdq_values} reports representative values for $d=5,10,20,$ and $50$.

\begin{table}[t]
\centering
\small
\setlength{\tabcolsep}{5pt}
\begin{tabular}{c cccc}
\toprule
& \multicolumn{4}{c}{$t_{\rm{finite}}$} \\
\cline{2-5}
$q$ & $d=5$ & $d=10$ & $d=20$ & $d=50$ \\
\midrule
  1 & 0 & 0  & 0  & 0 \\
 2 & 1 & 1  & 1  & 1 \\
 3 & 2 & 2  & 2  & 2 \\
 4 & 2 & 3  & 3  & 3 \\
 5 & 2 & 3  & 4  & 4 \\
 6 & 3 & 4  & 4  & 5 \\
 7 & 3 & 4  & 5  & 6 \\
 8 & 3 & 5  & 5  & 6 \\
 9 & 3 & 6  & 6  & 6 \\
10 & 3 & 6  & 7  & 7 \\
11 & 4 & 7  & 8  & 8 \\
12 & 4 & 6  & 9  & 9 \\
13 & 5 & 7  & 9  & 10 \\
14 & 4 & 8  & 10 & 11 \\
15 & 4 & 9  & 11 & 12 \\
16 & 5 & 9  & 12 & 13 \\
17 & 4 & 9  & 13 & 14 \\
18 & 5 & 10 & 13 & 15 \\
19 & 5 & 10 & 13 & 15 \\
20 & 5 & 11 & 14 & 15 \\
21 & 5 & 11 & 15 & 16 \\
22 & 5 & 12 & 16 & 17 \\
23 & 5 & 11 & 15 & 18 \\
24 & 5 & 12 & 16 & 19 \\
25 & 5 & 10 & 16 & 20 \\
\bottomrule
\end{tabular}
\caption{Representative exact values of $t_{\rm{finite}}$ computed from the linear-independence parameter $\rho$ of the upper-left $q\times q$ Sobol generating matrices for dimensions $d=5,10,20,50$}
\label{tab:tdq_values}
\end{table}

\subsection{Empirical ansatz and fitted parameter}
\label{app:empirical_ansatz_fit}

The exact computation of $t_{\rm{finite}}$ requires explicit rank checks over the binary generating matrices. To obtain a simple closed-form approximation suitable for the theoretical analysis in the main text, we fit the proposed empirical ansatz defined in \cref{eq:t_q_approx}. The parameter $\gamma$ is fitted by minimizing the squared error between the exact values $t_{\rm{finite}}$ and the approximation $\tilde t$ over the computed grid of $\mathcal{D}=(d,q)$. The fitted value is
\begin{equation}
    \gamma = 0.77215,
\end{equation}
And the resulting root-mean-square error is
\begin{equation}
    \operatorname{RMSE}
    =
    \left[
    \frac{1}{|\mathcal{D}|}
    \sum_{(d,q)\in \mathcal{D}}
    \left(
    t_{\rm{finite}}-\tilde t
    \right)^2
    \right]^{1/2}
    =
    0.911,
    \label{eq:tq_fit_rmse}
\end{equation}
Equivalently, this corresponds to an average squared error of approximately
$\operatorname{MSE} \simeq 0.830$.
Although the true finite-resolution values $t_{\rm{finite}}$ are not exactly monotone in $q$ for all dimensions, the ansatz captures the dominant saturation trend: for small $q$, the quality parameter increases approximately linearly with the resolution, while for larger $q$ it approaches the construction-dependent limiting value $t_{\rm{Sobol}}$.

\begin{figure}[t]
    \centering
    \includegraphics[width=1.\textwidth]{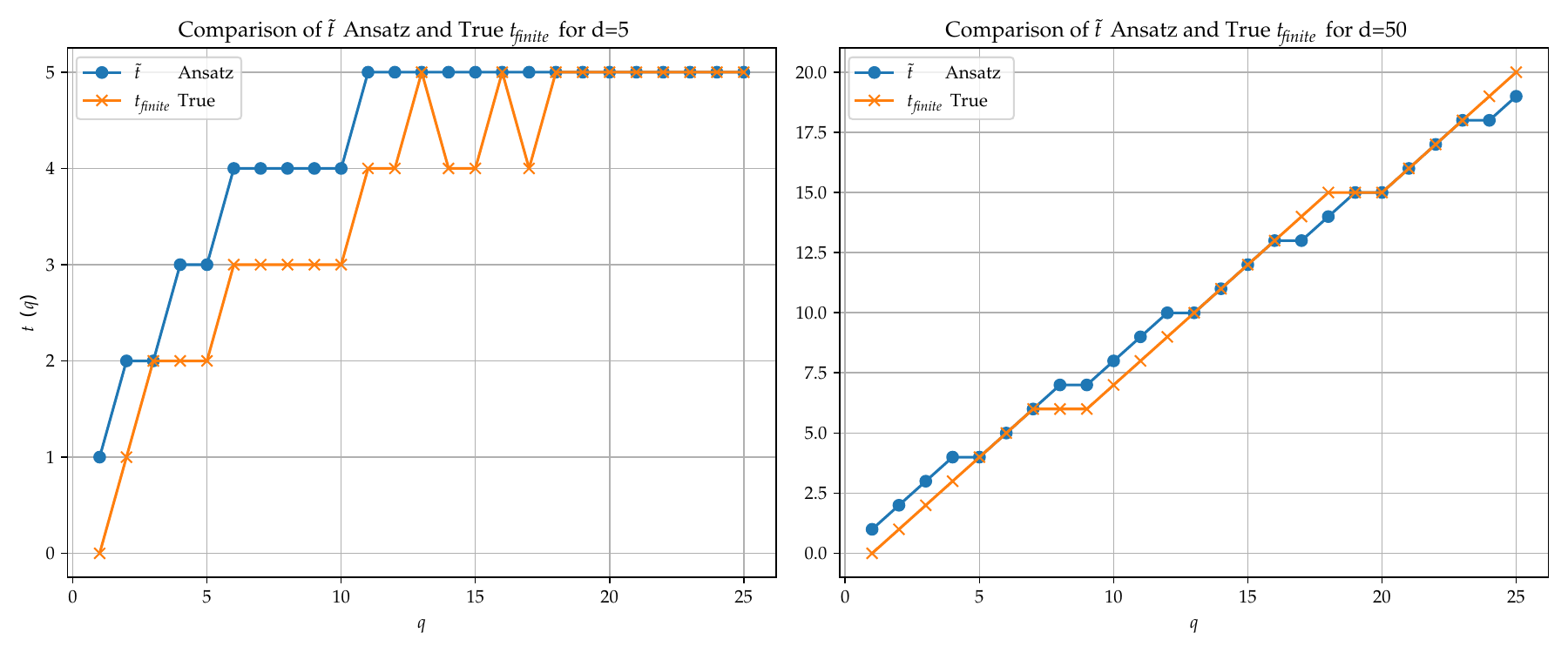}
    \caption{Comparison between the exact values $t_{\rm{finite}}$ and the empirical approximation $\tilde t$ for $d=5$ and $d=50$. The same fitted value $\gamma=0.77215$ is used for all dimensions.}
    \label{fig:tq_fit_examples}
\end{figure}

The agreement shown in \cref{fig:tq_fit_examples} indicates that a single global value of $\gamma$ already provides a reasonable approximation of the finite-resolution quality parameter across the tested dimensions.\footnote{One could fit a different value of $\gamma$ for each dimension $d$, which would generally improve the pointwise agreement with the exact values $t_{\rm{finite}}$. However, the purpose of \cref{eq:t_q_approx} is to provide a simple and interpretable empirical model for use in the theoretical analysis. Since a single value of $\gamma$ already captures the observed trend reasonably well, we keep one global parameter for all dimensions.}

\section{Evidence of pre-asymptotic quantum advantage in QqMC}
\label{app:proof_QqMC}

In this Appendix, we provide the proofs necessary to conduct the pre-asymptotic analysis in \cref{sec:QqMC}. Throughout this Appendix, we exploit the following notations 

\begin{equation}
\label{eq:estimators_def}
\begin{aligned}
&\tilde f \coloneq\int_{[0,1)^d} f(u)\,du,  \qquad\text{true integral value},\\
&f_q \coloneq\frac{1}{2^q}\sum_{i=1}^{2^q}f(x_i), \qquad\text{qMC estimator using $2^q$ Sobol points},\\
&f_Q \coloneq\frac{1}{2^Q}\sum_{i=1}^{2^Q}f(x_i), \qquad \text{Sobol estimator represented in the QqMC circuit},\\
&\hat f_Q(T) \coloneq \quad \text{QAE estimate of $f_Q$ using $T$ Grover-operator queries}.
\end{aligned}
\end{equation}
The corresponding errors are 
\begin{equation}
\label{eq:error_definitions}
\begin{aligned}
\epsilon_{qMC} &\coloneq|f_q-\tilde f|, &\text{classical qMC error with $2^q$ Sobol points},\\
\epsilon_{QqMC} &\coloneq |\hat f_Q(T)-\tilde f|, &\text{total QqMC error}.
\end{aligned}
\end{equation}

\begin{theorem}[Pre-asymptotic upper-bound advantage window for QqMC]
\label{th:Tthreshold}
Let use the notation introduced in \cref{eq:basic_qmc_notation,eq:Tmin_def}, and in \cref{eq:estimators_def,eq:error_definitions}. Assume $V(f)>0$. The upper bound for the proposed QqMC error $\epsilon_{QqMC}$ is strictly lower than the upper bound for the classical qMC error $\epsilon_{qMC}$ whenever the tuple $(T,q,l)$ satisfies
\begin{equation}
\label{eq:Tthreshold}
\boxed{
\Delta>0
\qquad\text{and}\qquad
T\ge T_{\min}(q,l,d,V(f)),
}
\end{equation}
with $T_{\min}$ defined in \cref{eq:Tmin_def}.

\end{theorem}

\begin{proof}
Using \cref{eq:error_definitions} and the triangle inequality, one obtains
\begin{equation}
\label{eq:triangle_QqMC}
\epsilon_{QqMC}
=
|\hat f_Q(T)-\tilde f|
\le
|\hat f_Q(T)-f_Q|+|f_Q-\tilde f|.
\end{equation}
By the QAE error bound in \cref{eq:aerr} and the Sobol qMC bound in \cref{eq:final_err_qmc_sob}, this gives
\begin{equation}
\label{eq:errbQqmc}
\epsilon_{QqMC}
\le
a(T)+V(f)\tilde g (Q)
=
a(T)+V(f) \tilde g(q+l).
\end{equation}
Always from~\cref{eq:errbQqmc}, the classical qMC error with $2^q$ Sobol points satisfies
\begin{equation}
\label{eq:classical_qmc_bound_q}
\epsilon_{qMC}
\le
V(f)\tilde g(q).
\end{equation}
Therefore, the QqMC upper bound is strictly smaller than the classical qMC upper bound whenever
\begin{equation}
\label{eq:basiccomparison}
a(T)+V(f) \tilde g(q+l)<V(f) \tilde g(q).
\end{equation}
Using the definition of $\Delta$ in \cref{eq:Deltaq_def}, this is equivalent to
\begin{equation}
\label{eq:aT_less_Delta}
a(T)<\Delta.
\end{equation}
If $\Delta\le0$, \cref{eq:aT_less_Delta} cannot hold because $a(T)>0$. Hence we must require $\Delta>0$. Under this condition, \cref{eq:aT_less_Delta} becomes
\begin{equation}
\label{eq:quadratic_T_condition}
\frac{\pi}{T}+\frac{\pi^2}{T^2}<\Delta.
\end{equation}
Multiplying by $T^2>0$ gives
\begin{equation}
\label{eq:quadratic_T_positive}
\Delta T^2-\pi T-\pi^2>0.
\end{equation}
Since $\Delta>0$, this quadratic inequality holds precisely when $T$ is strictly larger than the positive root,
\begin{equation}
\label{eq:positive_root_T}
T>
\frac{\pi+\sqrt{\pi^2+4\pi^2\Delta}}{2\Delta}
=\frac{2\pi}{\sqrt{1+4\Delta}-1}.
\end{equation}
Because $T$ is an integer, the smallest admissible value is exactly \cref{eq:Tmin_def}. This proves the claim.
\end{proof}

For the proof of the next theorem it is useful to introduce the following notations 

\begin{equation}
\label{eq:nz_phid_def}
n_z\coloneq z-\tilde t(z),
\qquad
\phi(n)\coloneq2^{-n}\sum_{i=0}^{d-1}\binom{n}{i},
\end{equation}
so that
\begin{equation}
\label{eq:gd_phid_relation}
\tilde g (z)=\phi(n_z).
\end{equation}
In this appendix, for the sake of clarity, we made the dependence of $t$ on $q$ explicit.
\begin{theorem}[Positivity of $\Delta$]
\label{th:Delta_sign}
Fix $d,l\in\mathbb N$. For any $V(f)\ge0$, the difference $\Delta$ defined in~\cref{eq:Deltaq_def} is always non-negative:
\begin{equation}
\label{eq:Delta_nonnegative}
\Delta\ge0 \qquad\forall q\ge1.
\end{equation}

Furthermore, if $V(f)>0$, strict positivity ($\Delta>0$) is achieved if and only if both of the following conditions hold:
\begin{align}
    & n_{q+l} \ge d \\
    & n_{q+l} > n_q
\end{align}
Equivalently, the second condition can be written directly in terms of the ansatz as $\tilde t(q+l)-\tilde t(q)<l$.
\end{theorem}

\begin{proof}
By \cref{eq:gd_phid_relation} and \cref{eq:Deltaq_def}, we have $\Delta = V(f)\bigl(\phi(n_q)-\phi(n_{q+l})\bigr)$. We first show that $n_q$ is nondecreasing. For $d \le 3$, this is immediate from the exact formulas for $\tilde t(q)$. For $d \ge 4$, the unrounded ansatz $h(q) = t_{\rm{Sobol}}\left(1-e^{-\gamma q/ t_{\rm{Sobol}}}\right)$ grows by strictly less than 1 at each step:
\begin{equation}
h(q+1)-h(q) = t_{\rm{Sobol}} e^{-\gamma q/t_{\rm{Sobol}}} \left(e^{-\gamma/t_{\rm{Sobol}}} - 1\right) < \frac{1}{\gamma} < 1.
\end{equation}
Because the continuous function grows by less than 1, its rounded counterpart $\tilde t(q)$ can increase by at most 1. Consequently, the effective trials $n_{q+1} - n_q = 1 - (\tilde t(q+1) - \tilde t(q))$ must be either 0 or 1, proving $n_q$ is nondecreasing.

Next, we analyze $\phi(n)$. From Pascal's identity, the backward difference is:
\begin{equation}
\label{eq:phid_difference}
\phi(n)-\phi(n+1) = 2^{-n-1}\binom{n}{d-1} \ge 0.
\end{equation}
This reveals that $\phi(n)$ is flat (i.e. $\phi(n)=1$) for $0 \le n \le d-1$, and strictly decreasing for $n \ge d$. Because $n_{q+l} \ge n_q$ and $\phi$ is nonincreasing, $\Delta \ge 0$ universally, which proves \cref{eq:Delta_nonnegative}.

If $V(f)>0$, the difference $\phi(n_q)-\phi(n_{q+l})$ is strictly positive exactly when the two arguments are distinct ($n_{q+l} > n_q$) and the larger argument has entered the strictly decreasing region of the function ($n_{q+l} \ge d$). Substituting the definition of $n_q$ into the inequality $n_{q+l} > n_q$ yields the equivalent condition $\tilde t(q+l)-\tilde t(q)<l$.
\end{proof}

\cref{th:Delta_sign} establishes that while $\Delta$ is never negative, strict positivity is required to make the QqMC upper bound strictly smaller than the classical qMC bound (otherwise \cref{eq:aT_less_Delta} fails for finite $T$). It is therefore critical to identify the earliest level $q$ where this strict inequality triggers.

\begin{corollary}[First strict-feasibility level]
\label{cor:qmin}
Define $q_{\min}(d, l)$ as the smallest integer $q \ge 1$ for which $\Delta > 0$. 

Let $m_d$ denote the threshold integer where the number $n_m$ first reaches $d$, i.e.
\begin{equation}
\label{eq:md}
    m_d \coloneq \min\{m\ge 1: n_m \ge d\}.
\end{equation}
Because the condition for strict positivity depends on $q+l$, we simply step back by $l$ from this threshold. If this falls below the physical limit $q=1$, strict positivity begins immediately at $q=1$. Thus:
\begin{equation}
\label{eq:qmin_cases}
q_{\min}(d, l) = 
\begin{cases} 
m_d - l & \text{if } m_d - l \ge 1, \\
1 & \text{otherwise}.
\end{cases}
\end{equation}

Explicitly, this yields:
\begin{itemize}
    \item For $d=1$ and $d=2$, the threshold is $m_d = d$, so $q_{\min} = 1$.
    \item For $d=3$, using the special branch, the threshold is $m_3=4$, giving $q_{\min} = \max\{1, 4-l\}$.
    \item For $d \ge 4$, the threshold $m_d$ is evaluated computationally as the first integer $m \ge 1$ satisfying $m - \tilde t(m) \ge d$.
\end{itemize}
\end{corollary}

\begin{proof}
By \cref{th:Delta_sign}, $\Delta>0$ if and only if $n_{q+l}\ge d$ and $n_{q+l}>n_q$. By the definition of the threshold $m_d$, the condition $n_{q+l} \ge d$ is satisfied exactly when $q+l \ge m_d$. Therefore, no $q < m_d - l$ can be strictly feasible.

If $m_d - l \ge 1$, evaluating at $q = m_d - l$ means $q+l = m_d$, fulfilling the first condition ($n_{q+l} \ge d$). Simultaneously, because $q < m_d$, we know $n_q \le d-1$. This guarantees $n_{q+l} > n_q$, fulfilling the second condition. If $m_d - l < 1$, the first available integer is $q=1$, which naturally satisfies $q+l \ge m_d$, and the strict increase follows from the same monotonicity argument. The explicit cases follow by direct substitution of the ansatz for $d \le 3$.
\end{proof}

For the sake of the pre-asypmtotic analysis carried in \cref{sec:feasibility_window_quantum_advantage}, it is worth investigate the asymptotic behavior of $q_{\min}$ as $d$ grows. The trend, given in \cref{eq:qmin_asymptotic_d} can be proven by first exploring how $m_d$ (defined in \cref{eq:md}), grows with $d$. Here the following lemma
\begin{lemma}[Asymptotic Behavior of $m_d$]
\label{lemma:m_d}
Let $\gamma \in (0, 1)$ be a constant. Let $t_{\rm{Sobol}} = O(d \log_2 d)$. Consider the ansatz given in \cref{eq:t_q_approx}. Let $m_d$ be the unique integer solution to the equation:
\begin{equation}
m - \tilde t(m) = d.
\end{equation}
Then, as $d \to \infty$, $m_d$ admits the following asymptotic expansion:
\begin{equation}
m_d = \frac{d}{1 - \gamma} + O\left(\ \frac{d}{\log d}\right).
\end{equation}
\end{lemma}

\begin{proof}
We seek an approximation for $m$ satisfying $m - t_{\rm{Sobol}}(1 - e^{-m\gamma/t_{\rm{Sobol}}}) = d$. Let us define the variable $x = \frac{m\gamma}{t_{\rm{Sobol}}}$. Since $t_{\rm{Sobol}} =O (d \log d)$, then $x \to 0$ as $d \to \infty$. This allows us to use the Taylor expansion of the exponential function around zero. Replacing this expansion up to second order into the expression for $\tilde t(m)$:
\begin{equation}
\tilde t(m) = m\gamma - \frac{m^2\gamma^2}{2t_{\rm{Sobol}}} + O\left(\frac{m^3}{t_{\rm{Sobol}}^2}\right).
\end{equation}
Now, substitute $\tilde t(m)$ back into the original equation $m - \tilde t(m) = d$, and ignoring the higher-order terms involving $1/t_{\rm{Sobol}}$, we obtain the leading order estimate $m_0$:
\begin{equation}
m_0 = \frac{d}{1 - \gamma}.
\end{equation}
To find the next term in the expansion, let $m = m_0 + \delta$, where $\delta$ is a correction term expected to be smaller than $m_0$. Substituting $m \approx m_0$ into the quadratic term $\frac{m^2\gamma^2}{2t_{\rm{Sobol}}}$ gives:
\begin{equation}
\frac{m^2\gamma^2}{2t_{\rm{Sobol}}} \approx \frac{m_0^2\gamma^2}{2t_{\rm{Sobol}}} = \frac{d^2 \gamma^2}{2t_{\rm{Sobol}}(1 - \gamma)^2}.
\end{equation}
Recalling that $t_{\rm{Sobol}} = O(d \log d)$, this term simplifies to:
\begin{equation}
\frac{m^2\gamma^2}{2t_{\rm{Sobol}}} \approx  \frac{d \gamma^2}{2(1 - \gamma)^2 \log d}.
\end{equation}
Thus, the equation becomes approximately:
\begin{equation}
m(1 - \gamma) + \frac{d \gamma^2}{2(1 - \gamma)^2 \ln d} \approx d.
\end{equation}
Solving for $m$:
\begin{equation}
m_d \approx \frac{d}{1 - \gamma} - \frac{d \gamma^2}{2(1 - \gamma)^3 \ln d}.
\end{equation}
This establishes the claimed asymptotic expansion. The uniqueness of the solution follows from the strict monotonicity of the function $u(m) = m - \tilde t(m)$ for $\gamma < 1$, since $u'(m) = 1 - \gamma e^{-m\gamma/t_{\rm{Sobol}}} > 1 - \gamma > 0$.
\end{proof}

\end{document}